\documentclass[a4paper,UKenglish,cleveref, autoref, thm-restate]{lipics-v2019}
\pdfoutput=1
\usepackage{amsmath}
\usepackage{bm}
\usepackage{amsthm}
\usepackage{amssymb}
\usepackage{comment}
\usepackage{stmaryrd}
\usepackage{mathtools}

\usepackage{tikz}
\usetikzlibrary{backgrounds, shapes.geometric, decorations.pathreplacing, patterns}
\pgfdeclarelayer{edgelayer}\pgfdeclarelayer{nodelayer}\pgfsetlayers{background,edgelayer,nodelayer,main}
\tikzstyle{bcomonoid}=[minimum width=.3cm,draw=black,inner sep=0.5pt,fill=black,minimum height=.3cm,font=\footnotesize,shape=diamond]
\tikzstyle{anycomonoid}=[pattern color=black,minimum height=.3cm,minimum width=.3cm,draw=black,inner sep=0.5pt,preaction={fill=white},pattern=north west lines,font=\footnotesize,shape=diamond]
\tikzstyle{hmonoid}=[minimum width=.2cm,draw=black,inner sep=0.5pt,fill=black,minimum height=.2cm,font=\footnotesize,shape=circle]
\tikzstyle{square}=[fill=white,shape=square]
\tikzstyle{anymonoid}=[pattern color=black,minimum height=.3cm,minimum width=.3cm,draw=black,inner sep=0.5pt,preaction={fill=white},pattern=north west lines,font=\footnotesize,shape=diamond]
\tikzstyle{none}=[]
\tikzstyle{monoid}=[minimum width=.3cm,draw=black,inner sep=0.5pt,fill=white,minimum height=.3cm,font=\footnotesize,shape=diamond]
\tikzstyle{gmonoid}=[minimum width=.2cm,draw=black,inner sep=0.5pt,fill=white,minimum height=.2cm,font=\footnotesize,shape=circle]
\tikzstyle{default_label}=[label position=right]
\tikzstyle{comonoid}=[minimum width=.3cm,draw=black,inner sep=0.5pt,fill=white,minimum height=.3cm,font=\footnotesize,shape=diamond]
\tikzstyle{bmonoid}=[minimum width=.3cm,draw=black,inner sep=0.5pt,fill=black,minimum height=.3cm,font=\footnotesize,shape=diamond]
\makeatletter\pgfkeys{/tikz/stretch/.code 2 args=\tikz@stretch{#1}{#2}}\def\tikz@stretch#1#2{\pgfpointtransformed{\pgfpoint{0cm}{0cm}}\pgf@xa=\pgf@x\pgf@ya=\pgf@y\pgfpointtransformed{\pgfpoint{#1cm}{#2cm}}\advance\pgf@x by-\pgf@xa\ifdim\pgf@x<0pt\pgf@x=-\pgf@x\fi\advance\pgf@y by-\pgf@ya\ifdim\pgf@y<0pt\pgf@y=-\pgf@y\fi\pgfkeysalso{/tikz/minimum width/.expanded=\the\pgf@x}\pgfkeysalso{/tikz/minimum height/.expanded=\the\pgf@y}}\makeatother

\hideLIPIcs

\newcommand{\ket}[1]{|#1\rangle}
\newcommand{\quot}[2]{{\raisebox{.2em}{$#1$}\left/\raisebox{-.2em}{$#2$}\right.}}

\newcommand{\TODO}[1]{%
\vskip .1cm%
\begingroup\par%
\hbox to \hsize{\strut\llap{\bf\large $\bigstar$ \`A faire%
$\bigstar$\quad}\vrule\hfil\parbox{.9\hsize}{\it #1}\hfil}
\par\endgroup%
\vskip .1cm%
}%

\makeatletter
\newcommand\mytag[2][]{%
  \def\@currentlabel{#2}%
  #2\label{#1} 
}
\makeatother

\title{A recipe for quantum graphical languages} 

\titlerunning{A recipe for quantum graphical languages} 

\author{Titouan Carette}{Universit\'e de Lorraine, CNRS, Inria, LORIA, F 54000 Nancy, France}{titouan.carette@loria.fr}{https://orcid.org/0000-0002-1825-0097}{}

\author{Emmanuel Jeandel}{Universit\'e de Lorraine, CNRS, Inria, LORIA, F 54000 Nancy, France}{emmanuel.jeandel@loria.fr}{https://orcid.org/0000-0002-1825-0097}{}

\authorrunning{T. Carette and E. Jeandel} 

\Copyright{Titouan Carette and Emmanuel Jeandel} 



\funding{ANR-17-CE25-0009, PIA-GDN/Quantex}

\EventEditors{Artur Czumaj, Anuj Dawar, and Emanuela Merelli}
\EventNoEds{3}
\EventLongTitle{47th International Colloquium on Automata, Languages, and Programming (ICALP 2020)}
\EventShortTitle{ICALP 2020}
\EventAcronym{ICALP}
\EventYear{2020}
\EventDate{July 8--11, 2020}
\EventLocation{Saarbrücken, Germany (virtual conference)}
\EventLogo{}
\SeriesVolume{168}
\ArticleNo{118}

\newtheorem{disc}[theorem]{Discussion}

\begin{document}
	\nolinenumbers
\maketitle

\begin{abstract}
  Different graphical calculi have been proposed to represent quantum computation. First the ZX-calculus \cite{coecke2011interacting}, followed by the ZW-calculus \cite{hadzihasanovic2015diagrammatic} and then the ZH-calculus \cite{backens2018zh}. We can wonder if new $Z^*$-calculi will continue to be proposed forever. This article answers negatively. All those language share a common core structure we call $Z^*$-algebras. We classify $Z^*$-algebras up to isomorphism in two dimensional Hilbert spaces and show that they are all variations of the aforementioned calculi. We do the same for linear relations and show that the calculus of \cite{bonchi2017interacting} is essentially the unique one.
\end{abstract}
	
	\vspace{0.5cm}
	
	
The most common formalization of quantum computing is the circuit
model, a diagrammatical language representing unitary matrices in a
two dimensional Hilbert space, see \cite{nielsen2002quantum} for an
introduction. Verification of quantum processes requires a sound and
complete equational theory for quantum circuits, i.e. a complete presentation of unitaries by generators and relations. This is known to be a difficult open problem.
	
By relaxing the unitarity condition and allowing all linear
maps, at least three different complete equational theories were found. The \emph{$ZX$-calculus} was introduced in
\cite{coecke2011interacting} and was designed as a part of the
categorical quantum mechanics program. It relies on the interaction
between two complementary observables.
The $ZX$-calculus has proven to be a good language to reason about
quantum processes \cite{de2020zx,duncan2019graph}. However, finding a set of rules
to make it complete has been open for a long time, and part of the
solution~\cite{jeandel2018complete} involved a secondary graphical language: the $ZW$-calculus
\cite{hadzihasanovic2015diagrammatic,coecke2010compositional}. This
calculus is built on two tripartite entanglement classes (GHZ and W-states) unraveling new
structures. Yet another complete graphical language was later
introduced, the $ZH$-calculus \cite{backens2018zh}, inspired by
hyper-graph states.

Compared to quantum circuits, these three languages share an important advantage. Processes and matrices are not represented merely by
diagrams, but by \emph{graphs} (hence the
term \emph{graph}ical language). Isomorphic
graphs represent the same quantum evolution. This peculiarity is embedded in the
\textit{only topology matters} paradigm. This is a subtle
feature: a usual diagrammatic language (like quantum circuits) starts
with a given set of primitives (usually quantum gates) for which the
notion of inputs and outputs is significant.
When \textit{only topology matters}, one can readily switch an input into an
output, and conversely.


\vspace{0.2cm}

This property follows from some specificities of the building blocks of those
languages. One goal of this article is to give a formal definition of these specificities.

Then, we will be able to prove that the three existing
graphical calculi for quantum computing, $ZX$, $ZH$ and $ZW$, are \emph{essentially} the only possible graphical calculi
for quantum computing.

To do this, we identify in this paper a common structure underlying
the already defined  calculi, that we call a
$Z^*$-algebra. Formally, the structure consists in two Frobenius
algebras interacting \textit{via} a bialgebra rule. To this, we add one
additional property, called \emph{compatibility}, to ensure the
\textit{only topology matters} paradigm. We then describe all the
$Z^*$-algebras in two dimensional Hilbert spaces and show that they all
happen to be phase-shifted versions of four structures we call $ZX$,
$ZH$, $ZW$ and $ZZ$. The first three appear respectively in the $ZX$,
$ZH$ and $ZW$-calculus. The last one corresponds to a degenerate calculus arising from only one self-interacting special Frobenius algebra.
	
It is important to note that languages for quantum computing are not
the only known to enjoy these nice properties. In particular
Bonchi and his coauthors  \cite{bonchi2017interacting} gave in 2017 a
graphical language for linear relations, with striking similarities to 
the $ZX$, $ZH$ and $ZW$ calculi. In fact, we will prove that their
language is \emph{essentially} the only graphical language for linear relations.

\vspace{.2cm}
There exist some other formalisms trying to unify graphical
languages, in particular in the context of interacting Frobenius
algebras \cite{duncan2016interacting} or Hopf-Frobenius algebras
\cite{collins2019hopf}. However, these formalisms usually require too
much structures, and fail to capture all three examples simultaneously. Typically they do not capture the $ZW$-calculus.

Some of our work indirectly has to do with the classification of
finite dimensional algebras, bialgebras and Frobenius algebras.
In the general case, an exact
classification of algebras is not known, even in the commutative
case. It is known that there is an infinite number of algebras up to
isomorphism of dimension $d$ for any $d > 6$. All of them are
known for $d\leq 6$ \cite{poonen2008isomorphism}. We find a classification of low dimension bialgebras in \cite{dekkar2008bialgebra}. We can find some constructions related to $Z^*$-algebras in \cite{koppinen1996algebras} and \cite{doi2000bifrobenius}.

	
\vspace{0.2cm}
	This paper starts by introducing the \emph{prop formalism} for graphical languages. The second section introduces various algebraic structures culminating in the definition of $Z^*$-algebras. The third section provides a classification of $Z^*$-algebras up to isomorphism for qubits. The last section gives some hint towards classification in higher dimension and provides a classification of $Z^*$-algebras in the context of linear relations.

    \section{Diagrammatical quantum computing}
	
In this paper all processes are represented by a combinatorial
structure called a prop \cite{zanasi2018interacting}. 
	
	\begin{definition}[prop]
		A \textbf{prop} $\textbf{P}$ is a collection of sets
                $\textbf{P}[n,m]$, indexed by $\mathbb{N}^2$. An
                element $f \in \textbf{P}[n,m]$ is called a morphism
                and is usually written
                $f:n\to m$. These sets are linked by the following
                operators:\vspace{0.25cm}
		\begin{itemize}
			\item A \textbf{composition} $\circ: \textbf{P}[b,c]\times \textbf{P}[a,b]\to \textbf{P}[a,c]$ satisfying: $(f\circ g)\circ h=f\circ (g\circ h)$.\vspace{0.25cm}
			\item A \textbf{tensor product} $\otimes: \textbf{P}[a,b]\times \textbf{P}[c,d]\to \textbf{P}[a+c,b+d]$, satisfying: $(f\otimes g)\otimes h=f\otimes (g\otimes h)$ and $(f\circ g)\otimes (h\circ k)=(f\otimes h)\circ (g\otimes k)$.\vspace{0.25cm}
			\item An \textbf{empty morphism} $1:0\to 0$
                          such that $f\otimes 1= 1\otimes f=f$ for all
                          $f:a\to b$.\vspace{0.25cm}
			\item An \textbf{identity} $id: 1\to 1$ such that
                          $f\circ id^{\otimes a}=id^{\otimes b} \circ
                          f=f$ for all $f:a\to b$. With the convention $id^{\otimes 0}=1$.
\vspace{0.25cm}
			\item A \textbf{symmetry} $\sigma: 2\to 2$
                          satisfying: $\sigma^2=id^{\otimes 2}$ and such that, $\sigma_{a}\circ (f\otimes id)=(id\otimes f)\circ \sigma_b$, for all $f:a\to b$, where $\sigma_{n+1}=(1^{\otimes n}\otimes \sigma)\circ (1\otimes \sigma_n)$, with $\sigma_0 = id$.
                        \end{itemize}
		
		\end{definition}
	
	In the language of categories \cite{MacLane}, a prop is a small strict symmetric monoidal category whose monoid of object is spanned by a unique object.
They can be seen as resource sensitive Lawvere theories where multiple
outputs are allowed \cite{bonchi2018deconstructing}.

	Props admit a nice diagrammatical representation that gives a topological interpretation to the axioms \cite{selinger2010survey}. A morphism $f:n\to m$ is represented as a box with $n$ inputs and $m$ outputs. Composition is represented by plugging the boxes. The tensor product by drawing the boxes side by side. The identity is represented by a single wire, the empty morphism by an empty diagram and the symmetry by wire crossing.
	  
	\begin{center}
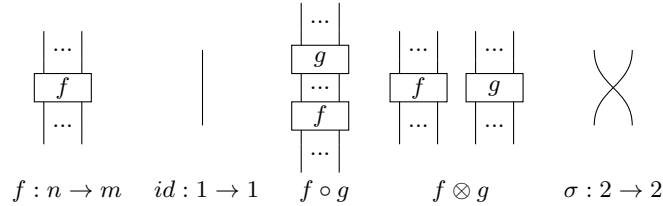
		\begin{tabular}{ccccc}
\input{figs/boxf} & \begin{tikzpicture}[baseline=(current bounding box.center)]
\begin{pgfonlayer}{nodelayer}
\coordinate (0) at (0.25,0);
\coordinate (1) at (0.25,0.5);
\coordinate (2) at (0.25,1);

\end{pgfonlayer}
\begin{pgfonlayer}{edgelayer}
\draw[] (0.center) to[out= 90, in=-89.9] (1.center) to[out= 90, in=-89.9] (2.center);

\end{pgfonlayer}
\end{tikzpicture}
 & \input{figs/boxfog} & \input{figs/boxftg} & \input{figs/swap}\\[0.5cm]
			$f:n\to m$ & $id:1\to 1$ & $f\circ g$ & $f\otimes g$ & $\sigma : 2\to 2$
		\end{tabular}
	\end{center}
	
	That choice of notation fits nicely with the axioms of props. The corresponding equations are natural in the diagrammatic notation. 
	In particular the symmetry axioms express that the boxes can move through wires:
	
	\begin{center}
		$\input{figs/sinv}\qquad\qquad\qquad \input{figs/snat}$
	\end{center}	
	This diagrammatical language is sound \cite{joyal1991geometry}. So we can equivalently work with equations or diagrams.

        \begin{example}[sets and functions]
          Let $X$ be a set. In the prop $\textbf{Fun}_X$ , the set
          $\textbf{Fun}_X[n,m]$ is exactly the set of functions from
          $X^n$ to $X^m$, with composition being the usual composition, and
          tensor product being the cartesian product. 
        \end{example}
	\begin{example}[matrices]
          For an integer $d$ and a field  $\mathbb{K}$,  the prop
        $\textbf{Mat}_{\mathbb{K}^d}$ is defined by $\textbf{Mat}_{\mathbb{K}^d}[n,m]\coloneqq
        \mathcal{M}_{d^m \times d^n} (\mathbb{K})$, the matrices of size $d^m$ by
        $d^n$ over the field $\mathbb{K}$.  The composition is the matrix product and the tensor is the Kronecker product.
Keeping with quantum computing traditions, we will denote by
$(\ket{e_i})_{1\leq i\leq d}$ a basis of $\mathbb{K}^d$.
	\end{example}

	The main prop of interest for quantum computing is $\textbf{Qubits}\coloneqq \textbf{Mat}_{\mathbb{C}^2}$. The quantum analog of bits, the qubits, are described by vectors in $\mathbb{C}^2$. A register of $n$ qubits is then a vector in the tensor product $\mathbb{C}^{2^n}$. 	
	
\section{Graphical structures}

While the diagrammatical languages presented in the previous section
make reasoning about props easier, it is still somewhat strict: inputs
come to the top of the box representing $f$, outputs goes out at the bottom.
\emph{Graphical languages} do not have this restriction, and we will
explain here what additional properties should be satisfied to obtain
a better framework.

\subsection{Half a spider}
We start by studying elementary associative binary operations with units: monoids.
	
\begin{definition}

A monoid is a morphism $\mu: 2 \rightarrow 1$ (the
product) and a morphism $\eta: 0 \rightarrow 1$ (the unit) that
satisfy the equations: ${\mu\circ (\eta\otimes \mathrm{id})=\mu\circ (\mathrm{id}\otimes \eta)=\mathrm{id}
 }$ and
${\mu\circ (\mu\otimes \mathrm{id})=\mu\circ (\mathrm{id}\otimes \mu)
 }$. If we depict the product
{\input figs/mprod } and the unit\vspace{-3mm}
{\begin{tikzpicture}[baseline=(current bounding box.center)]
\begin{pgfonlayer}{nodelayer}
\coordinate (0) at (0.25,0);
\node [style=monoid] (1) at (0.25,0.5) {};

\end{pgfonlayer}
\begin{pgfonlayer}{edgelayer}
\draw[] (0.center) to[out= 90, in=-89.9] (1.center);

\end{pgfonlayer}
\end{tikzpicture}
 }, the equations becomes:
${\input figs/muniteq }$ and
${\input figs/massoc }$\kern -1.5mm. The monoid
is commutative if ${\mu\circ \sigma=\mu
 }$. In pictures,
{\input figs/mcomeq }.
\end{definition}

\textbf{\emph{All the monoids in this paper are supposed to be commutative.}}

	Once we have a monoid $(\mu,\eta)$, we can define an $n$-ary
        product inductively by ${\mu_0 = \eta}, {\mu_1 = id}$ and
        ${\mu_{n+1} =  \mu \circ (\mu_n \otimes id)}$. As an example, here is $\mu_4$: {\input figs/m4 }.
	
	Using the equations, we have more generally $\mu_{n+p} = \mu \circ (\mu_n \otimes \mu_p)$, so how to
	transform the operator $\mu_n$ into compositions of $\mu$ and $\eta$ doesn't matter.

        \begin{example}
          A monoid in $\textbf{Fun}_X$ is exactly what is usually
          called a monoid on $X$. Monoid in $\textbf{Mat}_{\mathbb{K}^d}$
          are exactly the $d$-dimensional $\mathbb{K}$-algebras.
        \end{example}          

        In the following we will be mainly interested in the two
        following examples:
	\begin{example}[co-copy]
		Given a basis $\left(\ket{i}\right)_{1\leq i\leq d}$ of $\mathbb{K}^d$, the \textbf{co-copy} monoid is define in $\textbf{Mat}_{\mathbb{K}^d}$ by $\eta:1\mapsto \sum_{i=1}^{d}\ket{i}$ and $\mu:\ket{i}\ket{j}\mapsto \begin{cases}
		\ket{i}\text{ if } i=j\\
		\text{else } 0
		\end{cases}$
	\end{example}
	
	\begin{example}[monoid algebra \cite{Ponizovskii}]\label{ex}
		Given a monoid $M = \left(X,*,e\right)$ in
                $\textbf{Fun}_X$ with $X$ of cardinality $d$, we can define a monoid $\mathbb{K}[M]$ in $\textbf{Mat}_{\mathbb{K}^d}$ by indexing each element of a basis by the elements of $M$.
		We then take: $\eta: 1\mapsto \ket{e}$ and
                $\mu:\ket{a}\ket{b}\mapsto \ket{a*b}$.
               If $M$ is a group, we will speak of a \emph{group algebra}.
                
               Starting from a monoid $M$ in $\textbf{Fun}_X$ with $X$ of cardinality
$d+1$ that contains a
zero element (that we note $\bot$), we can build a \emph{contracted algebra}
$\mathbb{K}M$ in $\textbf{Mat}_{\mathbb{K}^d}$ by essentially the same
construction, but identifying the element $\bot$ with the matrix $0$.
One can see that the previous example of the co-copy actually  fits in
this framework: it is exactly $\mathbb{K}M$ for the monoid in
$\textbf{Fun}_{\{\bot, 1, \dots,n\}}$ defined by $i * j = i$ if $i=j$ and $\bot$ otherwise.

\end{example}
	
Any commutative monoid defines a group of phases:	
	\begin{definition}[phase]
		Given a commutative monoid $\left(\mu,\eta\right)$, a
		\textbf{phase} is an invertible morphism $\alpha:1\to 1$ such that: $\alpha\circ \mu=\mu\circ (\mathrm{id}\otimes \alpha)
 $. Pictorially: {\input figs/phase }.
\end{definition}
	
	The phases form an abelian group. In general we will write this group multiplicatively and write
	$\alpha \beta$ instead of $\alpha \circ \beta$.
	\emph{In the following, the notations $\alpha$ and $\beta$ will be reserved for elements of the phase group}.

\begin{disc}
  \label{disc:scalar}
An invertible scalar (a $0 \to 0$ morphism) is obviously a phase.
  Therefore the group of invertible scalars $S$ is always a subgroup of the phase group $G$.
  If it is a direct summand, i.e. if $G = S \times H$ for some group $H$, then one can simplify the
  presentation by ``dropping out'' the scalars and only consider nontrivial phases.
  This will be the case later on for the Qubit prop, but there are examples for which such a
  simplification cannot be made.
\end{disc}
	Once we have phases, we can introduce a new sequence of operators
$\mu_n(\alpha)$ and $\eta(\alpha)$ defined by $\mu_n(\alpha) = \alpha \circ
        \mu_n$ and $\eta(\alpha) = \alpha \circ \eta$ . Pictorially
{\input figs/m4ph } and
{\input figs/uph }.
        
	These operators satisfy the following equation: {\input figs/m4psh }

        It is interesting to note that $\mu(\alpha)$ itself defines a
        monoid, by taking as unit $\eta(\alpha^{-1})$. We will call it
        a \textbf{phase-shifted} monoid of the original monoid.
      
	Monoids dualize to co-monoids:	
	\begin{definition}
		A co-monoid in a prop is a morphism $\Delta: 1 \rightarrow 2$ (the co-product)
		and a morphism $\epsilon: 1 \rightarrow 0$ (the co-unit) that satisfies the
		equations: $(\Delta\otimes \mathrm{id})\circ \Delta=(\mathrm{id}\otimes \Delta)\circ \Delta
 $ and $(\epsilon\otimes \mathrm{id})\circ \Delta=(\mathrm{id}\otimes \epsilon)\circ \Delta=\mathrm{id}
 $. If we depict the co-product {\input figs/comprod }
\vspace{-4mm}
		and the co-unit {\begin{tikzpicture}[baseline=(current bounding box.center)]
\begin{pgfonlayer}{nodelayer}
\node [style=monoid] (0) at (0.25,0.25) {};
\coordinate (1) at (0.25,0.75);

\end{pgfonlayer}
\begin{pgfonlayer}{edgelayer}
\draw[] (0.center) to[out= 90, in=-89.9] (1.center);

\end{pgfonlayer}
\end{tikzpicture}
 }, the equations become:		
${\input figs/comassoc }$ and${\input figs/comuniteq }$.
		A co-monoid is co-commutative if it satisfies:
                $\sigma\circ \Delta=\Delta
 $. In pictures, {\input figs/comcomeq }.
              \end{definition}
\textbf{\emph{All the co-monoids in this paper are supposed to be cocommutative.}}
              
	Again we can inductively define $\Delta_n$ by $\Delta_0 = \epsilon, \Delta_1 = id$ and  $\Delta_{n+1} = (\Delta_n \otimes id)  \circ  \Delta$.
	
	We define phases for co-commutative co-monoids in the same way as phases for commutative monoids, they are the invertible morphisms satisfying: {\input figs/cophase }.
	
	We can also define the morphisms $\Delta(\alpha)$ and $\epsilon(\alpha)$ as well as the phase-shifted co-monoid $\left(\Delta(\alpha),\epsilon(1/\alpha)\right)$.

	\begin{example}[copy in $\textbf{Fin}_X$]
          The functions $\Delta: x \mapsto (x,x)$, with $\epsilon$ the
          only function from $X$ to $X^0$, defines a co-monoid in
          $\textbf{Fin}_X$. This is the only co-monoid in this prop.
	\end{example}
	
	\begin{example}[copy in $\textbf{Mat}_{\mathbb{K}^d}$]
		Given a basis of $\left(\ket{i}\right)_{1\leq i\leq d}$, the \textbf{copy} co-monoid is defined in $\textbf{Mat}_{\mathbb{K}^d}$ by $\epsilon: \ket{i}\mapsto 1$ and $\Delta:\ket{i}\mapsto \ket{i}\ket{i}$.
	\end{example}
	
	\begin{example}[group co-algebra]
		Given a finite group $G$ of size $d$ we can define a co-monoid in $\textbf{Mat}_{\mathbb{K}^d}$ by $\ket{x}\mapsto \frac{1}{d}\sum_{a* b=x}\ket{a}\ket{b}$, the co-unit is $\ket{x}\mapsto \begin{cases}
		1\text{ if } x=e\\
		\text{else } 0
		\end{cases}$.
	\end{example}

	\clearpage
	\subsection{One spider}
	
	A monoid and a co-monoid can interact forming a Frobenius algebra.
	
	\begin{definition}
		A monoid $(\mu,\eta)$ and a co-monoid $(\Delta, \epsilon)$
		form a Frobenius algebra iff they satisfy: ${(\mathrm{id}\otimes \mu)\circ (\Delta\otimes \mathrm{id})=(\mu\otimes \mathrm{id})\circ (\mathrm{id}\otimes \Delta)=\Delta\circ \mu
 }$. Pictorially: {\input figs/frobenius }.
	\end{definition}

A Frobenius algebra is commutative if the monoid is
        commutative and the co-monoid is cocommutative. \emph{\textbf{All the Frobenius algebras in this paper are commutative.}}

In a Frobenius algebra the phases of the monoid coincide with the phases of the co-monoid. Thus we can speak without ambiguity of the phases of a Frobenius algebra.
\begin{example}[In $\mathbf{Fin}_{X}$]
  There are no Frobenius algebras in $\mathbf{Fin}_{X}$ (unless $|X|=1$).
	\end{example}
	
	\begin{example}[copy and cocopy]
		Given a basis $\left(\ket{i}\right)_{1\leq i\leq d}$ of $\mathbb{K}^d$, the co-copy monoid and copy co-monoid form a Frobenius algebra in $\textbf{Mat}_{\mathbb{K}^d}$.
	\end{example}
	
	\begin{example}[group Frobenius algebra]
		Given a group $G$ of size $d$, the group algebra and the group co-algebra form a Frobenius algebra in $\textbf{Mat}_{\mathbb{K}^d}$.
	\end{example}
	
	When we have a Frobenius algebra $(\mu,\eta,\Delta, \epsilon)$ we can
	define a family of morphisms $S_{n,m}: n \rightarrow m$  by $S_{n,m} \coloneqq \mu_n \circ
	\Delta_m$. We call them \emph{spiders} and depict them {\input figs/s }.\vspace{-7mm}
	They satisfy the following equation: {\input figs/spider }. As we have done for monoids and co-monoids, provided a phase
        $\alpha$, we define decorated spiders $S_{n,m}(\alpha)$ by
        $S_n = m_n \circ \alpha \circ \Delta_p$. These new morphisms
        satisfy the equation:
{\input figs/spiderp }.

	\subsubsection{Compact structure}
	
	The symmetry in a prop allows various topological moves involving the wires. We can go further by providing a way to bend them. This is done by compact structures.
	
	\begin{definition}[compact structure]
		
A compact structure is given by two morphisms
$\delta:0\to 2$ and $\nu: 2\to 0$ depicted as
{\begin{tikzpicture}[baseline=(current bounding box.center)]
\begin{pgfonlayer}{nodelayer}
\node [] (0) at (0.5,0) {};
\coordinate (1) at (0.25,0.5);
\coordinate (2) at (0.75,0.5);

\end{pgfonlayer}
\begin{pgfonlayer}{edgelayer}
\draw[] (0.center) to[out=180, in=-89.9] (1.center);
\draw[] (0.center) to[out=  0, in=-89.9] (2.center);

\end{pgfonlayer}
\end{tikzpicture}
 } and
{\begin{tikzpicture}[baseline=(current bounding box.center)]
\begin{pgfonlayer}{nodelayer}
\coordinate (0) at (0.25,0);
\coordinate (1) at (0.75,0);
\node [] (2) at (0.5,0.5) {};

\end{pgfonlayer}
\begin{pgfonlayer}{edgelayer}
\draw[] (0.center) to[out= 90, in=180] (2.center);
\draw[] (1.center) to[out= 90, in=  0] (2.center);

\end{pgfonlayer}
\end{tikzpicture}
 } satisfying
the snake equation $ (\nu\otimes \mathrm{id})\circ (\mathrm{id}\otimes \delta)=(\mathrm{id}\otimes \nu)\circ (\delta\otimes \mathrm{id})=\mathrm{id}
 $. Pictorially:
{\input figs/snake }.                
		A compact structure is \emph{symmetric} if $\nu\circ \sigma=\nu\circ (\mathrm{id}\otimes \mathrm{id})
 $: pictorially, {{\input figs/scup }.} (This implies a similar statement on
                $\delta$). \emph{\textbf{All compact structures in
                    this paper are symmetric.}}
	\end{definition}
	
	A compact structure allows to bend the wire leading to new topological properties. This extends the diagrammatical language \cite{selinger2010survey}. Any Frobenius algebra directly provides a compact structure given by $ \delta=\Delta\circ \eta
 $ and $ \nu=\epsilon\circ \mu
 $, pictorially: {\input figs/cap } and
	{ \input figs/cup }.

	If the Frobenius algebra is commutative then this compact structure is symmetric.
	
	This compact structure behaves well with the Frobenius algebra, we have: $(\mathrm{id}\otimes \mu)\circ (\delta\otimes \mathrm{id})=\Delta
 $, pictorially: { \input figs/flex }
	
	This equation is interesting from a topological point of view. Bending the wires of a diagram gives a diagram representing
	the same morphism. This has been referred to as the \textit{only topology matters} paradigm \cite{coecke2011interacting}. \textbf{For us, the only topology matter paradigm is the key property of a graphical language}.
	
	In particular, we can by abuse of notation write: {\input figs/right }
	which may represent any of the following  diagrams : {\input figs/right2 }
	
	In general, we can give an unambiguous meaning to any multi-graph with input and outputs. We emphasize that this property plays a central role in the elegance of the $Z^*$ calculi. 
        
	\subsection{Two spiders}
	
	The ZX, ZW and ZH-calculii all have two Frobenius algebras. In fact a language based on only one spider is not expressive enough. The next step is therefore to have two of them.

	In this setting, the \textit{only topology matters} paradigm 
	doesn't apply anymore. Indeed, coloring the two algebras in white and black, we have: {\input figs/flex } and
        {\input figs/flex1 } using the two compact
        structures corresponding to the two algebras, 
	but in general {\input figs/notflex }.

So we cannot hope for the two compact structures to be equal,
but we can hope for some sort of \emph{compatibility}:
        
	\begin{definition}[compatibility]
		Two Frobenius algebras are \textbf{compatible} if their compact structure
		satisfy: \raisebox{-0.5\height}{\input figs/compat}. We call the left hand side the \textbf{dualizer}.
	\end{definition}  
	
	Note that the snake equation(s) implies that the left hand side is always the inverse of the right hand side. In compatible case the dualizer is an involution.
	
When the two Frobenius algebras are compatible, we can adjust the
language so that we can bend wires on both structures. This is done at
the price of a slight modification of the second algebra. We
introduce now a new generator (represented by a black node) that
represents the dualizer, and introduce four new generators in place of
the original structures:

\begin{center}
  \input{figs/hmdf}\hfill\input{figs/hudf}\hfill\input{figs/hddf}\hfill\input{figs/hedf}\hfill
  and
  more generally  \input{figs/hs}
\end{center}

With the new generators, we succeed in obtaining a new language:
Indeed: we can now bend the wires of the new generator, and we keep a
form of the spider rule:

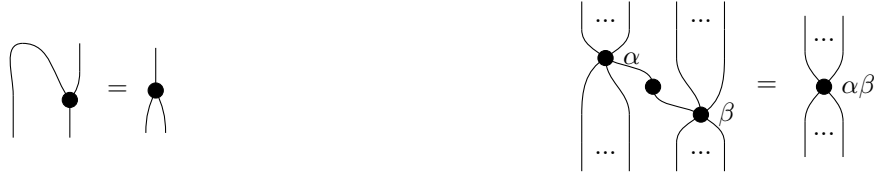
\begin{center}
  \input{figs/flex2}
  \hspace{5cm}
  \input{figs/hspiderp}
\end{center}

\begin{disc}One could decide similarly to change the first Frobenius algebra rather than the
second one. In fact, if there is a preexisting compact structure, it also make sense to search 
for a compatibility between the preexisting compact structure and the two algebras.
This is somehow what has been done in \cite{hadzihasanovic2015diagrammatic}.
\end{disc}

\subsection{Two spiders interacting}
	
	We now require the two spiders to interact in a precise way. 
	
	\begin{definition}[Bialgebra]
		A co-monoid and a monoid form a bialgebra iff they satisfy the three following equations:
		
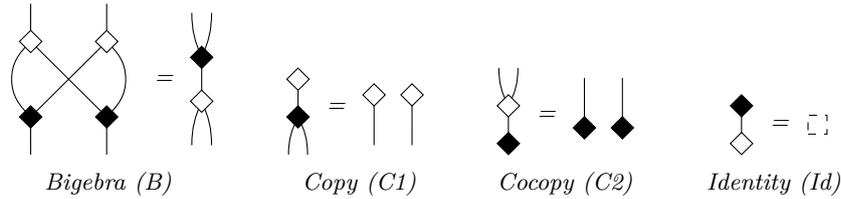
\begin{center}
			\begin{tabular}{ccccccc}
				\input{figs/bigebra} && \input{figs/copy} && \input{figs/cocopy} && \input{figs/id}\\
				Bigebra \mytag[bigebra]{(B)} &&Copy
                                \mytag[copy]{(C1)}&& Cocopy
                                \mytag[cocopy]{(C2)} && Identity \mytag[id]{(Id)}
			\end{tabular}
		\end{center}
	\end{definition}

The four bialgebra laws enforces some kind of commutation property
between the co-monoid and the monoid.
There are conflicting definitions in the literature on which properties one should impose on a
bialgebra. The one we take is from Sweedler\cite{sweedler}.

We now come to our main definition:
	\begin{definition}[$Z^*$-algebra]
		A \textbf{$Z^*$-algebra} is formed by two compatible
                Frobenius algebras such that the co-monoid of the first
                one satisfies the bigebra rule \ref{bigebra} with the monoid of the
                second one.
	\end{definition}

A $Z^*$-algebra formed by two Frobenius algebras $F$ and $G$ will be denoted $FG$.

\begin{disc}One could give a different definition of a $Z^*$-algebra, by imposing all four
conditions of the bialgebra law, or even impose it to both monoid/co-monoid pairs.
However it turns out that the most important examples (esp. the ZW-calculus) do not satisfy all
equations. We isolate the bigebra law as being central.\end{disc}

Using the notations from the previous section, we see that a
$Z^*$-algebra leads to a \emph{graphical}-calculus, formed by two
spiders that are subject to the following rules\footnote{The white
  node is the same as the white lozenge, but represented differently to
emphasize that the whole calculus is different}:

\hspace{-0.45cm}\input{figs/gspiderp} \hspace{0.1cm} \input{figs/hspiderp}\hspace{0.1cm} \input{figs/inv}\hspace{0.1cm} \input{figs/hbigebra} 

Together with the \textit{only topology matters} paradigm, which means we can
bend the wires of any node, changing an input into an output.

The rules we obtain are a common subset of the rules of the
$Z^*$-calculi
\cite{bonchi2017interacting,coecke2011interacting,hadzihasanovic2015diagrammatic,backens2018zh}.

\section[Classification of Z*-algebras in Qubits and LinRel]{Classification of $Z^*$-algebras in Qubits and LinRel}

Now that we have defined what we think is a graphical calculus, we can
proceed to the main theorem: there are essentially only four possible calculi for quantum
  computing up to isomorphism: the ZX-calculus, the ZW-calculus, the ZW-calculus, and
  the (trivial) ZZ-calculus.
Before we can give a formal statement of the theorem, we need to
explain what we mean by ``essentially''.

Consider a $Z^*$-algebra formed of two
Frobenius algebras named $A$ and $B$.
Suppose that $\lambda$ is a invertible scalar (i.e. a $0
\rightarrow 0$ morphism).
If we multiply, say, the generators of the monoid of $A$ by $\lambda$
and the generators of the co-monoid of $B$ by $1/\lambda$, then we
obtain a new $Z^*$-algebra. This new algebra is usually not isomorphic to the
first one, but for all practical purposes, they behave the same.

More generally, suppose we add a phase $\alpha$ to the monoid of $A$
(replacing $\mu, \eta$ by $\mu(\alpha)$ and $\eta(\alpha)$) and we
add similarly a phase $\beta$ to the co-monoid of $B$.
Then we obtain two new Frobenius algebras that we will call
$A^{\alpha}$ and $B_{\beta}$ which satisfies all axioms of a
$Z^*$-algebra, \emph{except possibly the compatibility relations}.
We call this a \emph{phase-shifted} versions of the original $Z^*$-algebra.
We will show that all possible graphical calculi for quantum computing
are phase-shifted version of four basic ones.

Phase-shifted algebras are a bit subtle. To ease the understanding, we provide here the graphical calculus that corresponds to  $A^{\alpha}B_{\beta}$ in terms of the original generators, with the caveat that it is a graphical calculus only if the compatibility relation is satisfied (equivalently, the black node below is an involution). $n$ and $m$ denote respectively the number of inputs and outputs and we represent the compact structure of the white algebra differently in both calculi to avoid confusion:

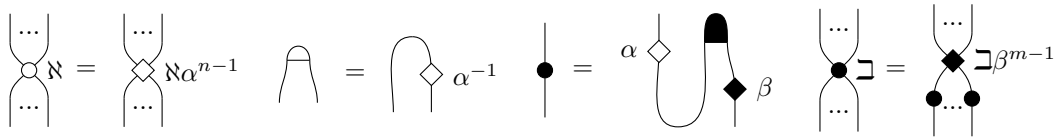
\begin{center}
\input figs/whiteps
\hfill
\input figs/capps
\hfill
\input figs/singleblackps
\hfill
\input figs/blackps
\end{center}

\subsection[Z*-algebras in Qubits]{$Z^*$-algebras in Qubits}

We now investigate the particular case of graphical calculi for
quantum computing. This corresponds to the special case $\mathbf{Qubits} = \textbf{Mat}_{\mathbb{C}^2}$.

A monoid in Qubit is exactly the same as a $\mathbb{C}$-algebra of
dimension 2.  Algebras in dimension $2$ have been
classified~\cite{study1890systeme}: there are only two algebras up to isomorphism. A proof is in Appendix \ref{studyproof}.

For our purpose however, we will introduce four algebras (the first three being isomorphic), that we call $Z$, $X$, $H$ and $W$. Working in the
basis $\left(\ket{0},\ket{1}\right)$. They correspond to contracted algebras $\mathbb{C}M$, see \ref{ex}.
	
	\begin{center}
		\begin{tabular}{|c|c|c|}
			\hline
			$Z$ & $\ket{0}$ & $\ket{1}$ \\
			\hline

			$\ket{0}$ & $\ket{0}$ & $0$ \\

			\hline
			$\ket{1}$ & $0$ & $\ket{1}$  \\
			\hline
		\end{tabular} $\quad$ \begin{tabular}{|c|c|c|}
		\hline
		$X$ & $\ket{0}$ & $\ket{1}$ \\
		\hline
		
		$\ket{0}$ & $\ket{0}$ & $\ket{1}$  \\
		
		\hline
		$\ket{1}$ & $\ket{1}$ & $\ket{0}$  \\
		\hline
		\end{tabular} $\quad$ \begin{tabular}{|c|c|c|}
\hline
$H$ & $\ket{0}$ & $\ket{1}$ \\
\hline

$\ket{0}$ & $\ket{0}$ & $\ket{0}$  \\

\hline
$\ket{1}$ & $\ket{0}$ & $\ket{1}$  \\
\hline
		\end{tabular} $\quad$ \begin{tabular}{|c|c|c|}
			\hline
			$W$ & $\ket{0}$ & $\ket{1}$ \\
			\hline
			
			$\ket{0}$ & $\ket{0}$ & $\ket{1}$  \\
			
			\hline
			$\ket{1}$ & $\ket{1}$ & $0$  \\
			\hline
		\end{tabular}
	\end{center}
	
	Those multiplication tables describe the behavior of the algebras on $\ket{0}$ and $\ket{1}$.
	
	We see that $Z$ behaves like a Kronecker delta ensuring equality, $X$ is the XOR gate, $H$ the AND gate and $W$ is the effect algebra on two elements.
	
	The matricial representation in the computational basis are:
	
	\begin{center}
          \input figs/mprodany
          ,
          \begin{tikzpicture}[baseline=(current bounding box.center)]
\begin{pgfonlayer}{nodelayer}
\coordinate (0) at (0.25,0);
\node [style=anymonoid] (1) at (0.25,0.5) {};

\end{pgfonlayer}
\begin{pgfonlayer}{edgelayer}
\draw[] (0.center) to[out= 90, in=-89.9] (1.center);

\end{pgfonlayer}
\end{tikzpicture}

          :=
	  \begin{tabular}{|c|c|}
            \hline
            \multicolumn{2}{|c|}{$Z$}\\
			\hline
			$\mu_Z$ & $\eta_Z$ \\
			\hline
			$\scalebox{0.5}{$\begin{pmatrix} 1&0&0&0\\ 0&0&0&1 \end{pmatrix}$}$
			&$\scalebox{0.5}{$\begin{pmatrix} 1\\1 \end{pmatrix}$}$\\
			\hline
		\end{tabular} $\quad$ \begin{tabular}{|c|c|}
            \hline
            \multicolumn{2}{|c|}{$X$}\\
		\hline
		$\mu_X$ & $\eta_X$ \\
		\hline
		$\scalebox{0.5}{$\begin{pmatrix} 1&0&0&1\\ 0&1&1&0 \end{pmatrix}$}$
		&$\scalebox{0.5}{$\begin{pmatrix} 1\\0 \end{pmatrix}$}$\\
		\hline
	  \end{tabular} $\quad$ \begin{tabular}{|c|c|}
            \hline
            \multicolumn{2}{|c|}{$H$}\\            
	\hline
	$\mu_H$ & $\eta_H$ \\
	\hline
	$\scalebox{0.5}{$\begin{pmatrix} 1&1&1&0\\ 0&0&0&1 \end{pmatrix}$}$
	&$\scalebox{0.5}{$\begin{pmatrix} 0\\1 \end{pmatrix}$}$\\
	\hline
          \end{tabular} $\quad$ \begin{tabular}{|c|c|}
\hline
\multicolumn{2}{|c|}{$W$}\\            
\hline
$\mu_W$ & $\eta_W$ \\
\hline
$\scalebox{0.5}{$\begin{pmatrix} 1&0&0&0\\ 0&1&1&0 \end{pmatrix}$}$
&$\scalebox{0.5}{$\begin{pmatrix} 1\\0 \end{pmatrix}$}$\\
\hline
\end{tabular}
	\end{center}

The phase group of $Z$ is ${\mathbb{C}_{\times}^*}^2$. The phase group of $W$ is $\mathbb{C}_{\times}^*\times \mathbb{C}_{+} $.

If we write the phases for our four favorite monoids, they read:
	
\begin{center}
\input figs/anyphase
:=  
		\begin{tabular}{|c|}
			\hline
			$Z$ \\
			\hline
			$a,b\in \mathbb{C}^*$\\
			\hline
			$\scalebox{0.5}{$a\begin{pmatrix} 1&0\\ 0&b\end{pmatrix}$}$\\
			\hline
		\end{tabular} $\quad$ \begin{tabular}{|c|c|}
			\hline
			$X$\\
			\hline
			$a,b\in \mathbb{C}^*$\\
			\hline
			$\scalebox{0.5}{$\frac{a}{2}\begin{pmatrix} 1+b&1-b\\ 1-b&1+b\end{pmatrix}$}$\\
			\hline
		\end{tabular} $\quad$ \begin{tabular}{|c|c|}
		\hline
		$H$\\
		\hline
		$a,b\in \mathbb{C}^*$\\
		\hline
		$\scalebox{0.5}{$a\begin{pmatrix} 1&1-b\\ 0&b\end{pmatrix}$}$\\
		\hline
		\end{tabular} $\quad$ \begin{tabular}{|c|c|}
			\hline
			$W$\\
			\hline
			$a\in \mathbb{C}^*$ $b\in\mathbb{C}$\\
			\hline
			$\scalebox{0.5}{$a\begin{pmatrix} 1&0\\ b&1\end{pmatrix}$}$\\
			\hline
		\end{tabular}
	\end{center}
        
As explained in Discussion~\ref{disc:scalar}, in the case of Qubit, we can write\footnote{This can be
done more generally in any prop if the group of scalars is divisible.} the phase group
$G = \mathbb{C}^\star \times H$ where the first component $\mathbb{C}^\star$ corresponds to the invertible scalars,
and $H$ is some commutative group ($H = \mathbb{C}^\star_{\times}$ in the first
three cases, and $H = \mathbb{C}_+$ in the last case). One could then index the phases only by this subgroup $H$ (i.e. always take $a = 1$), introducing scalars when necessary. This is what has been done in the literature.

All four monoids form Frobenius algebras, with the following co-monoids\footnote{The co-monoids have been choosen such that the three first Frobenius algebras are isomorphic, hence the weird $\frac{1}{2}$ factor in $X$.}, co-units, and compact structures:
	
\begin{center}
          \input figs/anyprod
          ,
          \begin{tikzpicture}[baseline=(current bounding box.center)]
\begin{pgfonlayer}{nodelayer}
\node [style=anymonoid] (0) at (0.25,0.25) {};
\coordinate (1) at (0.25,0.75);

\end{pgfonlayer}
\begin{pgfonlayer}{edgelayer}
\draw[] (0.center) to[out= 90, in=-89.9] (1.center);

\end{pgfonlayer}
\end{tikzpicture}

          :=
	  \begin{tabular}{|c|c|}
            \hline
            \multicolumn{2}{|c|}{$Z$}\\
            
			\hline
			$\Delta_Z$ & $\epsilon_Z$ \\
			\hline
			$\scalebox{0.5}{$\begin{pmatrix} 1&0\\ 0&0\\ 0&0\\0&1 \end{pmatrix}$}$
			&$\scalebox{0.5}{$\begin{pmatrix} 1&1 \end{pmatrix}$}$\\
			\hline
	  \end{tabular} $\quad$ \begin{tabular}{|c|c|}
            \hline
            \multicolumn{2}{|c|}{$X$}\\            
			\hline
			$\Delta_X$ & $\epsilon_X$ \\
			\hline
			$\scalebox{0.5}{$\frac{1}{2}\begin{pmatrix} 1&0\\ 0&1\\ 0&1\\1&0 \end{pmatrix}$}$
			&$\scalebox{0.5}{$\begin{pmatrix} 2&0 \end{pmatrix}$}$\\
			\hline
	  \end{tabular} $\quad$ \begin{tabular}{|c|c|}
            \hline
            \multicolumn{2}{|c|}{$H$}\\            
			\hline
			$\Delta_H$ & $\epsilon_H$ \\
			\hline
			$\scalebox{0.5}{$\begin{pmatrix} 1&2\\0&-1\\0&-1\\ 0&1 \end{pmatrix}$}$
			&$\scalebox{0.5}{$\begin{pmatrix} 1&2\end{pmatrix}$}$\\
			\hline
	  \end{tabular} $\quad$ \begin{tabular}{|c|c|}
            \hline
            \multicolumn{2}{|c|}{$W$}\\            
			\hline
			$\Delta_W$ & $\epsilon_W$ \\
			\hline
			$\scalebox{0.5}{$\begin{pmatrix} 0&0\\ 1&0 \\ 1&0 \\ 0&1 \end{pmatrix}$}$
			&$\scalebox{0.5}{$\begin{pmatrix} 0& 1 \end{pmatrix}$}$\\
			\hline
		\end{tabular}
	\end{center}

\begin{center}
\input figs/anycap
,
\input figs/anycup
:=  
	  \begin{tabular}{|c|c|}
            \hline
            \multicolumn{2}{|c|}{$Z$}\\
            
			\hline
			$\delta_Z$ & $\nu_Z$ \\
			\hline
			$\scalebox{0.5}{$\begin{pmatrix} 1\\ 0\\ 0\\1 \end{pmatrix}$}$
			&$\scalebox{0.5}{$\begin{pmatrix} 1&0&0&1 \end{pmatrix}$}$\\
			\hline
	  \end{tabular} $\quad$ \begin{tabular}{|c|c|}
            \hline
            \multicolumn{2}{|c|}{$X$}\\            
			\hline
			$\delta_X$ & $\nu_X$ \\
			\hline
			$\scalebox{0.5}{$\frac{1}{2}\begin{pmatrix} 1\\ 0\\ 0\\1 \end{pmatrix}$}$
			&$\scalebox{0.5}{$\begin{pmatrix} 2&0&0&2 \end{pmatrix}$}$\\
			\hline
	  \end{tabular} $\quad$ \begin{tabular}{|c|c|}
            \hline
            \multicolumn{2}{|c|}{$H$}\\            
			\hline
			$\delta_H$ & $\nu_H$ \\
			\hline
			$\scalebox{0.5}{$\begin{pmatrix} 2\\-1\\-1\\ 1 \end{pmatrix}$}$
			&$\scalebox{0.5}{$\begin{pmatrix} 1&1&1&2\end{pmatrix}$}$\\
			\hline
	  \end{tabular} $\quad$ \begin{tabular}{|c|c|}
            \hline
            \multicolumn{2}{|c|}{$W$}\\            
			\hline
			$\delta_W$ & $\nu_W$ \\
			\hline
			$\scalebox{0.5}{$\begin{pmatrix} 0\\ 1 \\ 1 \\ 0 \end{pmatrix}$}$
			&$\scalebox{0.5}{$\begin{pmatrix} 0& 1&1&0 \end{pmatrix}$}$\\
			\hline
		\end{tabular}
	\end{center}
        We can now state our main theorem:
\begin{theorem}\label{zalg}
		The only $Z^*$-algebras up to isomorphism in $\textbf{Qubits}$ are, with $a,b\in \mathbb{C}^*$: $ Z^{(a,\frac{b}{a})} Z_{(\frac{1}{a},\frac{a}{b})}$, $ Z^{(a,\frac{b}{a})} Z_{(-\frac{1}{a},\frac{a}{b})}$, $ Z^{(a,\frac{b}{a})} Z_{(\frac{1}{a},-\frac{a}{b})}$, $ Z^{(a,\frac{b}{a})} Z_{(-\frac{1}{a},-\frac{a}{b})}$, $ Z^{(a,1)} X_{(\frac{2}{a},1)}$, $ Z^{(a,1)} X_{(-\frac{2}{a},1)}$, $ Z^{(a,-1)} X_{(\frac{2}{a},1)}$, $ Z^{(a,-1)} X_{(-\frac{2}{a},1)}$, $Z^{(\frac{a}{b},b^2)} X_{(\frac{2}{a},-1)}$, $Z^{(a,-1)} X_{(\frac{2b}{a},\frac{1}{b^2})}$, $ Z^{(a,\frac{1}{b^2-1})} H_{(\frac{b}{a},\frac{1-b^2}{b^2})}$ with $b^2\neq 1$, $ Z^{(a,\frac{1}{b^2})} W_{(\frac{b}{a},0)}$ and $ W^{(a,0)} Z_{(\frac{b}{a},\frac{1}{b^2})}$.
\end{theorem}

The proof is detailed in Appendix \ref{zalgproof}. The idea is to show that, up to isomorphism, there
are only five possible monoids/co-monoid pairs
satisfying the bigebra rule, and then show how they can possibly
extend to $Z^*$-algebras.

We will now compare the calculi we obtain with the literature.

\begin{itemize}
  \item The $ZZ$-calculus never has been really considered, as having two
    spiders that are identical is not useful. However, its
    existence is not happenstance: in general a Frobenius algebra
    would not make a bigebra with itself. In this case, it works as
    $Z$ is a \emph{special} Frobenius algebra.
  \item The $ZX$-calculus \cite{coecke2011interacting} corresponds to
    what we call $ZX_{(2,2)}$. This is a particular calculus as the dualizer is trivial: both algebras have the same compact structure (up to scalars). We say that the two algebras are \emph{coflexible}. There are a few substantial
    differences between our calculus and the $ZX$-calculus. Instead of using all possibles phases in $\mathbb{C}^\star$, the authors use phases in the unit circle. Subsequent work \cite{ng2018completeness}
    introduced so-called lambda-boxes to restore all phases.
    Second, the $ZX_{(2,2)}$-calculus is a bit awkward as the two Frobenius algebras $Z$ and $X_{(2,2)}$ are not isomorphic, but only isomorphic up to a scalar. By rescaling the $X$ algebra, we can obtain a calculus where both algebras are dual, at the price of a slightly different bigebra rule. 
The isomorphism corresponds to the Hadamard matrix; as this matrix is symmetric, we can add it to our language without changing the \textit{only topology matters} paradigm, and we obtain this way the $ZX$-calculus defined in \cite{coecke2011interacting}.
  \item The fact that other calculi of the form $Z^\alpha X_\beta$ exist
corresponds to some commutation properties between phases of the two algebras. In fact, they correspond to what is called \emph{the
      $\pi$-commutation rule}: $ (1, \lambda)_Z \circ (1,-1)_X=\lambda (1,-1)_X \circ (1, \frac{1}{\lambda})_Z$ where $(a, b)_Z$ is a phase of $Z$ and $(a, b)_X$ is a phase of $X$.
  \item The original rules \cite{coecke2011interacting} of the $ZX$-calculus correspond exactly to:
    \begin{itemize}
    \item The \textit{only topology matters} paradigm (rule T)
    \item The rules  above valid on any graphical calculi (rules S1, S2, B2), including in this case
      the copy rule (B1) and some form of the identity rule (rule D1)
    \item one rule relative to the $\pi$-commutation (rule K2)
    \item one rule relative to Hadamard, the isomorphism between Z and X (rule C)
    \item one rule stating that $Z$ is a special Frobenius algebra (hidden in rule S1)
    \item one rule called $\pi$-copying (K1), and one rule related to the scalars (rules D2). The second rule is anecdotal. The first one  relates to what are the automorphisms of
      $Z$.
    \end{itemize}
    Therefore, with one omission, all original rules of the $ZX$-calculus could be rediscovered
    again in a systematic way using our definition.
  \item The $ZW$-calculus as discussed in \cite{hadzihasanovic2017algebra,HNW} is exactly
    what we call $ZW$. The calculus however do not use phases on the
    black nodes. The original $ZW$-calculus introduced in
    \cite{hadzihasanovic2015diagrammatic} by the same author is
    slightly different. Intuitively it corresponds to a different kind
    of graphical languages where the two Frobenius algebras have been
    made compatible with a third compact structure.
The fact that other calculi of the form $Z^\alpha W_\beta$ exist
essentially amounts to the same $\pi$-commutation rule as before.
  \item The $ZH$-calculus as discussed in \cite{backens2018zh} is
    exactly what we call $ZH_{(\sqrt{2}, -\frac{1}{2})}$. However the authors do not use
    phases on the white node, and use a different parametrizations of
    the phases on the black node. The phase they call $x$ is what
    we would call the phase $(1, 1-2x)_H$. This makes the spider rule more
    awkward in their calculus. The fact that other calculi of the form $Z^\alpha H_\beta$ exist is
    linked to the following rule: $ 2\frac{\lambda+1}{\lambda}(1, \lambda)_Z \circ H \circ (1,\frac{1}{2}\frac{1}{\lambda +1})_H = (1,2(\lambda +1))_H \circ H \circ  (1, \frac{1}{\lambda})_Z$ where $(a, b)_Z$ is a phase of $Z$, $(a, b)_H$ is a phase of $H$ and $H$ is the Hadamard gate.
  \end{itemize}  

\subsection{Generalization for qudits}

A similar classification could be theoretically done for other
dimensions, i.e. for the prop $\textbf{Mat}_{\mathbb{C}^d}$. However
difficulties arise. Indeed, all possible algebras have been
classified only in dimension $d \leq 6$ \cite{poonen2008isomorphism}
(in fact there are an infinite number of non isomorphic commutative
algebras of dimension $7$), and the work is even more terse on
bigebras (some work \cite{dekkar2008bialgebra} has been done for bi\emph{algebras} in dimension
$2$ and $3$) or Frobenius algebras (although a theoretical
characterization exist).

We will therefore focus here on generalizations of the existing
structures of dimension $2$ to higher dimension.

The $ZX$-structure  corresponds to an interaction between the two
algebras $\mathbb{C}^2$ and  $\mathbb{C}[\mathbb{Z}/2\mathbb{Z}]$.
One could readily generalize this to higher dimensions replacing $
\mathbb{Z}/2\mathbb{Z}$ by $\mathbb{Z}/d\mathbb{Z}$, or actually any
other commutative group of size $d$.
In higher dimension, the dualizer actually becomes non trivial: it
corresponds to the morphism $x \mapsto -x$ in
$\mathbb{Z}/d\mathbb{Z}$, which is trivial only if $d = 2$.

The $ZW$-structure can also be generalized easily. We again replace
the algebra $\mathbb{C}^2$ by $\mathbb{C}^d$. and we can generalize
the $W$ algebra in dimension $2$ to the algebra 
$\quot{\mathbb{C}[X]}{\left(X^d\right)}$.

We did not find any generalization of $ZH$ that work in arbitrary
dimension.
The obvious $d$-dimensional generalization of $H$ would be
as a contracted monoid algebra $\mathbb{C}M$ where $M$ is the monoid $({\cal
  F}, \cap)$ for a family of subsets ${\cal F}$ closed under
intersection (the $2$-dimensional version corresponding to ${\cal F} = \{ \{1\}, \emptyset\}$).
This generalization gives indeed two Frobenius algebra that satisfy a bigebra law, but they usually
are not compatible (unless ${\cal F} = 2^X$ for some set $X$).

\subsection[In LinRel]{In $\textbf{LinRel}_{\mathbb{K}}$}

Quantum computing is not the only place where graphical calculi appear: another $Z^*$-algebra that
occurs in the literature is \textit{Graphical linear algebra} \cite{bonchi2017interacting} in the
prop $\textbf{LinRel}_{\mathbb{K}}$. In $\textbf{LinRel}_{\mathbb{K}}$ a map $n\to m$ is a linear
subspace of $\mathbb{K}^{n+m}$.

It turns out that there are only two monoids in $\textbf{LinRel}_{\mathbb{K}}$, and they are not isomorphic:
the monoid given by the subspace $\{(x,x,x), x\in \mathbb{K}\}$ and the monoid given by $\{(x,y,x+y), x,y\in
\mathbb{K}\}$. Their respective phase groups are both trivial. Both these monoids, that we call $B$ and $N$, actually happen to have Frobenius algebra structures.
	
\begin{proposition}\label{monlin}
  There are only four $Z^*$-algebras in $\textbf{LinRel}_{\mathbb{K}}$: $BB$, $NN$, $BN$ and $NB$.
\end{proposition}

As these are the only potential candidates, we just have to check that they indeed give
$Z^*$-algebras. A detailed proof is given in Appendix \ref{monlinproof}.
$BB$ and $NN$ are trivial, and $BN$ is just a dual version of $NB$, therefore the graphical calculi
of \cite{bonchi2017interacting} is THE only possible graphical calculus for this prop.

\section{Future works}

We have classified $Z^*$-algebras in $\textbf{Mat}_{\mathbb{C}^2}$ and $\textbf{LinRel}_{\mathbb{K}}$. Further investigations will concern other categories.  In the case of $\textbf{Mat}_{\textbf{R}}$ for a semiring $\textbf{R}$, generalizations of $ZW$ and $ZX$ exist. A natural question is which other $Z^*$-algebras exist in this setting. All the monoids and co-monoids we considered were commutative, the non commutative case is also of interest, leading to a more general notion of graphical language involving \textit{port graphs} or \emph{rotation systems}. An other direction would be to drop the unit and the compact structure and find what defines a graphical language in this case. This is necessary in infinite dimensional Hilbert spaces for example.

\clearpage

	\bibliography{supp}
  
\appendix
	
	\section{Proofs}
	
	\subsection{Classification of two dimensional algebras} \phantomsection\label{studyproof}
	
		\begin{theorem}[\cite{study1890systeme}]
		In Qubit, any algebra is isomorphic either to $Z$ or to $W$.
	\end{theorem}
	
	\begin{proof}
		We are looking for all unital algebras up to isomorphism in $\mathbb{K}^2$, whith $\textbf{char}(\mathbb{K})\neq 2$.

	Given an algebra with unit, we choose a basis $\left(\ket{0},\ket{1}\right)$ where $\ket{0}$ is the unit. Then the matrix representation of the monoid is $\begin{pmatrix}
	1&0&0&x\\ 0&1&1&y
	\end{pmatrix}$. The change of basis $\begin{pmatrix}
	1&\frac{y}{2}\\ 0&1
	\end{pmatrix}$ gives $\begin{pmatrix}
	1&0&0&\lambda\\ 0&1&1&0
	\end{pmatrix}$ with $\lambda\coloneqq x+\frac{y^2}{2}$. Let $\begin{pmatrix}
	a&b\\ c&d
	\end{pmatrix}$ be an invertible matrix. Its determinant is $\Delta\coloneqq ad-bc\neq 0$. We want:\\
	
	\begin{center}
	$\begin{pmatrix}
	a&b\\ c&d
	\end{pmatrix}\begin{pmatrix}
	1&0&0&\lambda\\ 0&1&1&0
	\end{pmatrix}\begin{pmatrix}
	d&-b\\ -c&a
	\end{pmatrix}^{\otimes 2}=\Delta^2 \begin{pmatrix}
	1&0&0&\mu\\ 0&1&1&0
	\end{pmatrix}$
	\end{center}
	
	
	this gives the following system:
	
	\begin{center}
	$\begin{cases}
	\Delta=ad-bc\neq 0\\
	\Delta^2=ad^2-2bcd+\lambda ac^2\\
	0=c\left(b^2 - \lambda a^2\right)\\
	\Delta^2 \mu=\lambda a^3 -ab^2\\
	0=c\left(\lambda c^2 -d^2\right)\\
	\Delta^2=ad^2-\lambda ac^2\\
	0=b^2 c-2abd+\lambda a^2 c
	\end{cases}$
	\end{center}
	
	If $c\neq 0$ then we have $d^2=\lambda c^2$ and then $\Delta^2=0$, a contradiction. Setting $c=0$ the system reduces to:
	
	\begin{center}
	$\begin{cases}
	\Delta=ad\neq 0\\
	\Delta^2=ad^2\\
	\Delta^2 \mu=\lambda a^3 -ab^2\\
	0=-2abd\\
	c=0
	\end{cases}\Rightarrow \quad\begin{cases}
	\Delta=ad\neq 0\\
	\Delta^2=ad^2\\
	\Delta^2 \mu=\lambda a^3\\
	b=0\\
	c=0
	\end{cases}\Rightarrow \quad\begin{cases}
	d\neq 0\\
	\mu=\frac{\lambda}{d^2}\\
	a=1\\
	b=0\\
	c=0
	\end{cases}$
	\end{center}
		
	Finally we have $a=1$, $b=0$, $c=0$ and $d\neq 0$. The equivalence classes correspond to the elements of $\mathbb{K}$ up to multiplication by non-zero squares. We have three equivalence classes in $\mathbb{R}$: $\lambda< 0$, $\lambda> 0$ and $\lambda= 0$. In $\mathbb{C}$ there are only two $\lambda=0$ and $\lambda\neq0$. The case $\lambda\neq0$ admit a very simple representative: the change of basis $\frac{1}{\sqrt{2}}\begin{pmatrix}
	1&-1\\
	1&1
	\end{pmatrix}$ gives $\begin{pmatrix}
	1&0&0&0\\
	0&0&0&1
	\end{pmatrix}$.
	
	\end{proof}

	\subsection{Proof of Theorem \ref{zalg}} \phantomsection\label{zalgproof}
	 
	 To simplify our classification up to isomorphism, we start by identifying all the algebra automorphisms in $\textbf{Qubits}$. 
	 
	 \begin{proposition}
	 	The unique non-trivial automorphisms of $\mu_Z$ and $\mu_W$ are respectively $\begin{pmatrix} 0&1\\ 1&0\end{pmatrix}$ and the matrices of the form $ \begin{pmatrix} 1&0\\ 0&a\end{pmatrix}$ with $a\in\mathbb{C}^*$.
	 \end{proposition}
	 
	\begin{proof}
		
		We start with $\mu_Z$:
		
		\begin{center}
			$\begin{pmatrix}
			a&b\\ c&d
			\end{pmatrix}\begin{pmatrix}
			1&0&0&0\\ 0&0&0&1
			\end{pmatrix}\begin{pmatrix}
			d&-b\\ -c&a
			\end{pmatrix}^{\otimes 2}= \Delta^2\begin{pmatrix}
			1&0&0&0\\ 0&0&0&1
			\end{pmatrix}$
		\end{center}
		
		
		this gives the following system:
		
		\begin{center}
			$\begin{cases}
			\Delta=ad-bc\neq 0\\
			\Delta^2=ad^2+bc^2\\
			0=-ab\left(c+d\right)\\
			0=ab\left(a+b\right)\\
			0=cd\left(c+d \right)\\
			0=-cbd-dac\\
			\Delta^2=cb^2+da^2
			\end{cases}$
		\end{center}
		
		If $a=0$ then:
		
		\begin{center}
			$\begin{cases}
			\Delta=bc\neq 0\\
			\Delta^2=bc^2\\
			0=cd\left(c+d \right)\\
			0=-cbd\\
			\Delta^2=cb^2
			\end{cases}\Rightarrow \quad\begin{cases}
			\Delta=bc\neq 0\\
			\Delta^2=bc^2\\
			\Delta^2=cb^2\\
			d=0
			\end{cases} \Rightarrow\quad \begin{cases}
			d=0\\
			b=1\\
			c=1
			\end{cases}$
		\end{center}
		
		the solution is $\begin{pmatrix}0&1\\1&0
		\end{pmatrix}$. If $a\neq 0$ and $b\neq 0$ we then have $\Delta=0$, a contradiction. If $a\neq 0$ and $b=0$:
		
		\begin{center}
			$\begin{cases}
			\Delta=ad\neq 0\\
			\Delta^2=ad^2\\
			0=cd\left(c+d \right)\\
			0=-dac\\
			\Delta^2=da^2\\
			b=0
			\end{cases}\Rightarrow \quad\begin{cases}
			\Delta=ad\neq 0\\
			\Delta^2=ad^2\\
			\Delta^2=da^2\\
			c=0\\
			b=0
			\end{cases} \Rightarrow \quad\begin{cases}
			a=d=1\\
			c=0\\
			b=0
			\end{cases}$
		\end{center}
		
		the solution is $\begin{pmatrix}1&0\\0&1
		\end{pmatrix}$. Now for $\mu_W$:
		
		\begin{center}
			$\begin{pmatrix}
			a&b\\ c&d
			\end{pmatrix}\begin{pmatrix}
			1&0&0&0\\ 0&1&1&0
			\end{pmatrix}\begin{pmatrix}
			d&-b\\ -c&a
			\end{pmatrix}^{\otimes 2}= \Delta^2\begin{pmatrix}
			1&0&0&0\\ 0&1&1&0
			\end{pmatrix}$
		\end{center}
		
		
		this gives the system:
		
		\begin{center}
		$\begin{cases}
		\Delta=ad-bc\neq 0\\
		\Delta^2=ad^2-2bcd\\
		0=b^2c\\
		0=ab^2\\
		0=cd^2\\
		\Delta^2=ad^2\\
		0=cb^2-abd
		\end{cases} \Rightarrow \quad\begin{cases}
		\Delta=ad\neq 0\\
		\Delta^2=ad^2\\
		b=c=0
		\end{cases} \Rightarrow \quad\begin{cases}
		d\neq 0\\
		a=1\\
		b=c=0
		\end{cases}$
		\end{center}
		the solutions are the matrices $\begin{pmatrix}1&0\\0&d
		\end{pmatrix}$ with $d\neq 2$.
		
	\end{proof}
	
	This result allows to find all the monoid/co-monoid pair satisfying the \ref{bigebra} rule.
	
	\begin{lemma}
		In \textbf{Qubits}, up to isomorphism, the only monoid/co-monoid pair satifying the \ref{bigebra} rule are $\mu_Z/\Delta_Z$, $\mu_X/\Delta_Z$, $\mu_W/\Delta_Z$, $\mu_H/\Delta_Z$, and $\mu_Z/\Delta_W$.
	\end{lemma}

	\begin{proof}
		
		There are only two co-algebras up to isomorphism, $\Delta_Z$ and $\Delta_W$.
		
		Any algebra is of the form: $\begin{pmatrix}
		a&b&b&c\\ d&e&e&f
		\end{pmatrix}$.
		
		We start by finding all the algebras satisfying \ref{bigebra} with $\Delta_W$.
		
		We want: \begin{center}
		$\Delta_W \circ \begin{pmatrix}
		a&b&b&c\\ d&e&e&f
		\end{pmatrix}=\begin{pmatrix}
		a&b&b&c\\ d&e&e&f
		\end{pmatrix}^{\otimes 2}\left[I_2\otimes \begin{pmatrix}
		1&0&0&0\\ 0&0&1&0\\ 0&1&0&0\\ 0&0&0&1
		\end{pmatrix}\otimes I_2\right]{\Delta_W}^{\otimes 2}$
		\end{center}
		
		This gives the following system:
		
		\begin{center}
			\hspace{-0.5cm}$\begin{cases}
			0=a\left(a-1\right)\\
			0=d\left(a-1\right)\\
			0=d^2\\
			0=b\left(2a-1\right)\\
			0=e\left(a-1\right)+dc\\
			0= de\\
			0= c\left(2a-1\right) +2b^2\\
			0= f\left(a-1\right)+2be+cd\\
			0= 2df+2e^2
			\end{cases} \hspace{-1cm}\Rightarrow 
			\begin{cases}
			0=a\left(a-1\right)\\
			d=0\\
			0=b\left(2a-1\right)\\
			0= c\left(2a-1\right) +2b^2\\
			0= f\left(a-1\right)+2bc\\
			e= 0
			\end{cases} \Rightarrow
			\begin{cases}
			\text{ if $a=0$: } \begin{cases}
			\begin{tabular}{cc}
			$a=0$ & $b=0$\\
			$c=0$ & $d=0$\\
			$e=0$ & $f=0$\\
			\end{tabular}
			\end{cases}\\
			\\
			\text{ if $a\neq 0$: } \begin{cases}
			\begin{tabular}{cc}
			$a=1$ & $b=0$\\
			$c=0$ & $d=0$\\
			$e=0$ & $f\in \mathbb{C}$\\
			\end{tabular}
			
			\end{cases}\\
			\end{cases}$
		\end{center}
		
		The only rank $2$ solution are the $\begin{pmatrix}
		1&0&0&0 \\ 0&0&0&f
		\end{pmatrix}$ with $f\in \mathbb{C}^*$. They are algebras with units $\begin{pmatrix}
		1\\ \frac{1}{f}
		\end{pmatrix}$. Since $\begin{pmatrix}
		1&0\\0&f
		\end{pmatrix}$ is an automorphism of $\Delta_W$ this gives a unique pair up to isomorphism: $\mu_Z/\Delta_W$.
		
		Now with $\Delta_Z$, we want:
		\begin{center}
		$\Delta_Z \circ \begin{pmatrix}
		a&b&b&c\\ d&e&e&f
		\end{pmatrix}=\begin{pmatrix}
		a&b&b&c\\ d&e&e&f
		\end{pmatrix}^{\otimes 2}\left[I_2\otimes \begin{pmatrix}
		1&0&0&0\\ 0&0&1&0\\ 0&1&0&0\\ 0&0&0&1
		\end{pmatrix}\otimes I_2\right]{\Delta_Z}^{\otimes 2}$
		\end{center}
		
		This gives the following system:
		
		\begin{center}
		$\begin{cases}
		
		\begin{tabular}{ccc}
		$a=a^2 $&$ 0=ad $&$ d=d^2$\\
		$b=b^2 $&$ 0=be $&$ e= e^2$\\
		$c= c^2 $&$ 0= cf $&$ f= f^2$
		\end{tabular}
		\end{cases} \Leftrightarrow\quad
		\begin{cases}
		a,b,c,d,e,f\in \{0,1\}\\
		(a\neq 1)\lor (d\neq 1)\\
		(b\neq 1)\lor (e\neq 1)\\
		(c\neq 1)\lor (f\neq 1)
		\end{cases}$
		\end{center}
		
		The rank $2$ solutions are: 
		
		\begin{center}
			\begin{tabular}{cccc}
				
				$\begin{pmatrix}
				1&0&0&0 \\ 0&0&0&1
				\end{pmatrix}$ &
				$\begin{pmatrix}
				1&0&0&1 \\ 0&1&1&0
				\end{pmatrix}$ &
				$\begin{pmatrix}
				0&1&1&0 \\ 1&0&0&1
				\end{pmatrix}$ &
				$\begin{pmatrix}
				1&1&1&0 \\ 0&0&0&1
				\end{pmatrix}$ \\[0.5cm]
				
				$\begin{pmatrix}
				1&0&0&0 \\ 0&1&1&1
				\end{pmatrix}$&
				$\begin{pmatrix}
				1&0&0&0 \\ 0&1&1&0
				\end{pmatrix}$&
				$\begin{pmatrix}
				0&1&1&0 \\ 0&0&0&1
				\end{pmatrix}$&
				$\begin{pmatrix}
				0&0&0&1 \\ 1&0&0&0
				\end{pmatrix}$\\ [0.5cm]
				
				$\begin{pmatrix}
				0&1&1&1 \\ 1&0&0&0 
				\end{pmatrix}$&
				$\begin{pmatrix}
				0&0&0&1\\1&1&1&0 
				\end{pmatrix}$&
				$\begin{pmatrix}
				0&0&0&1 \\ 0&1&1&0
				\end{pmatrix}$& 
				$\begin{pmatrix}
				0&1&1&0 \\ 1&0&0&0
				\end{pmatrix}$
			\end{tabular}
		\end{center}
		
		Since $\begin{pmatrix}
		0&1\\1&0
		\end{pmatrix}$ is an automorphism of $\Delta_Z$, this reduces the possibilities to:\begin{center}
			\begin{tabular}{cccc}
				
				$\begin{pmatrix}
				1&0&0&0 \\ 0&0&0&1
				\end{pmatrix}$ &
				$\begin{pmatrix}
				1&0&0&1 \\ 0&1&1&0
				\end{pmatrix}$ &
				$\begin{pmatrix}
				1&1&1&0 \\ 0&0&0&1
				\end{pmatrix}$ &
				$\begin{pmatrix}
				1&0&0&0 \\ 0&1&1&0
				\end{pmatrix}$ \\[0.5cm]
				
				$\begin{pmatrix}
				0&0&0&1 \\ 1&0&0&0
				\end{pmatrix}$&
				$\begin{pmatrix}
				0&0&0&1\\1&1&1&0 
				\end{pmatrix}$& 
				$\begin{pmatrix}
				0&1&1&0 \\ 1&0&0&0
				\end{pmatrix}$
			\end{tabular}
		\end{center}
		
		But among them the last three are not algebras, they are not associative, a counter example for the three maps is the evaluation of $(\ket{0}*\ket{0})*\ket{1}$ versus $\ket{0}*(\ket{0}*\ket{1})$. The other are the algebras $\mu_Z$, $\mu_X$, $\mu_H$ and $\mu_W$.
		
		This gives $4$ pairs, $\mu_Z/\Delta_Z$, $\mu_X/\Delta_Z$, $\mu_W/\Delta_Z$ and $\mu_H/\Delta_Z$.
	\end{proof}
	
	Now we characterize all the possible Frobenius algebras given a fixed monoid or co-monoid.
	
		\begin{lemma}\label{famfrob}
		Given a commutative Frobenius algebra $F\coloneqq\left(\mu,\eta,\Delta,\epsilon\right)$, the co-monoids forming Frobenius algebras with $\left(\mu,\eta\right)$ are exactly the phase shifted co-monoids $\left(\Delta_\varphi,\epsilon_\varphi\right)$. $\left(\mu,\eta,\Delta_\varphi,\epsilon_\varphi\right)$ is called the phase shifted Frobenius algebra.
		\end{lemma}
	
	\begin{proof}
		$\left(\Rightarrow\right)$ Given a phase $\alpha$, the phase-shifted co-monoid $\left(\Delta(\alpha),\epsilon(-\alpha)\right)$ forms a Frobenius algebra with $\left(\mu,\eta\right)$ (just moving around the phases).\\
		$\left(\Leftarrow\right)$ Let $\left(\input figs/d1 , \begin{tikzpicture}[baseline=(current bounding box.center)]
\begin{pgfonlayer}{nodelayer}
\node [style=bmonoid] (0) at (0.25,0.25) {};
\coordinate (1) at (0.25,0.75);

\end{pgfonlayer}
\begin{pgfonlayer}{edgelayer}
\draw[] (0.center) to[out= 90, in=-89.9] (1.center);

\end{pgfonlayer}
\end{tikzpicture}
 \right)$ be a co-monoid forming a Frobenius algebra with $\left(\mu,\eta\right)$. We define the morphisms {\input figs/p2 } and {\input figs/p1 }. Those morphisms satisfy the phase equation: 
		\begin{center}
			{\input figs/phaseproof1 } and {\input figs/phaseproof2 }
		\end{center}
		Furthermore they are inverse of each other: 
		\begin{center}
			{\input figs/invproof1 } and {\input figs/invproof2 }
		\end{center}
		We call them $\alpha$ and $\alpha^{-1}$. Furthermore we have:

		\begin{center}
			{\input figs/pproof1 } and {\input figs/pproof2 }
		\end{center}
		
		Finally $\Delta'=\Delta(\alpha)$ and $\epsilon'=\epsilon(-\alpha)$.
	\end{proof}
	
	We are ready to classify the $Z^*$-algebras.
	
		\begin{theorem}
		The only $Z^*$-algebras up to isomorphism in $\textbf{Qubits}$ are, with $a,b\in \mathbb{C}^*$: $ Z^{(a,b)} Z_{(\frac{1}{a},\frac{1}{b})}$, $ Z^{(a,b)} Z_{(-\frac{1}{a},\frac{1}{b})}$, $ Z^{(a,b)} Z_{(\frac{1}{a},-\frac{1}{b})}$, $ Z^{(a,b)} Z_{(-\frac{1}{a},-\frac{1}{b})}$, $ Z^{(a,1)} X_{(\frac{2}{a},1)}$, $ Z^{(a,1)} X_{(-\frac{2}{a},1)}$, $ Z^{(a,-1)} X_{(\frac{2}{a},1)}$, $ Z^{(a,-1)} X_{(-\frac{2}{a},1)}$, $Z^{(a,\frac{4}{a^2 b^2})} X_{(b,-1)}$, $Z^{(a,-1)} X_{(b,\frac{4}{a^2 b^2})}$, $ Z^{(a,\frac{1}{a^2 b^2-1})} H_{(b,\frac{1-a^2 b^2}{a^2 b^2})}$ with $a^2 b^2\neq 1$, $ Z^{(a,\frac{1}{a^2b^2})} W_{(b,0)}$ and $ W^{(a,0)} Z_{(b,\frac{1}{b^2 a^2})}$.
	\end{theorem}
	
	\begin{proof}
		The candidate $Z^*$-algebras are $Z^\alpha Z_\beta$, $Z^\alpha X_\beta$, $Z^\alpha W_\beta$, $Z^\alpha H_\beta$ and  $W^\alpha Z_\beta$. We only need to check compatibility. If $d$ is  without phase-shift and $\alpha$ and $\beta$ are the phase we do the phase-shifts with, compatibility corresponds to the equation: $\beta\circ d\circ \alpha=\alpha^{-1} \circ d^{-1}\circ \beta^{-1}$.
		
		\begin{itemize}
			\item $Z^\alpha Z_\beta$: The dualizer of $ZZ$ is the identity. Let $\alpha=(a,b)$ and $\beta=(c,d)$, $a,b,c,d\in\mathbb{C}^*$. $Z^\alpha$ and $Z_\beta$ are compatible iff 
			
			\begin{center}
				$c\begin{pmatrix}
				1&0\\0&d
				\end{pmatrix}a\begin{pmatrix}
				1&0\\0&b
				\end{pmatrix}=\frac{1}{a}\begin{pmatrix}
				1&0\\0&\frac{1}{b}
				\end{pmatrix}\frac{1}{c}\begin{pmatrix}
				1&0\\0&\frac{1}{d}
				\end{pmatrix}$
			\end{center}
		
		This gives the system:
		
		\begin{center}
			$\begin{cases}
			a^2 c^2=1\\
			
			b^2 d^2=1
			\end{cases}$
		\end{center}
	
	The $Z^*$-algebras are then $ Z^{(a,b)} Z_{(\frac{1}{a},\frac{1}{b})}$, $ Z^{(a,b)} Z_{(-\frac{1}{a},\frac{1}{b})}$, $ Z^{(a,b)} Z_{(\frac{1}{a},-\frac{1}{b})}$ and $ Z^{(a,b)} Z_{(-\frac{1}{a},-\frac{1}{b})}$. The dualizer is the identity for $ Z^{(a,b)} Z_{(\frac{1}{a},\frac{1}{b})}$. 
			
			\item $Z^\alpha X_\beta$: The dualizer of $ZX$ is $\frac{1}{2}$, its inverse is $2$. let $\alpha=(a,b)$ and $\beta=(c,d)$, $a,b,c,d\in\mathbb{C}^*$. $Z^\alpha$ and $X_\beta$ are compatible iff 
			
			\begin{center}
				$c\begin{pmatrix}
				1+d&1-d\\1-d&1+d
				\end{pmatrix}\frac{a}{2}\begin{pmatrix}
				1&0\\0&b
				\end{pmatrix}=\frac{1}{a}\begin{pmatrix}
				1&0\\0&\frac{1}{b}
				\end{pmatrix}\frac{2}{c}\begin{pmatrix}
				1+\frac{1}{d}&1-\frac{1}{d}\\1-\frac{1}{d}&1+\frac{1}{d}
				\end{pmatrix}$
			\end{center}
			
			This gives the system:
			
			\begin{center}
				$\begin{cases}
				(da^2 c^2 -4)(1+d)=0\\
				(a^2 c^2 bd+4)(1-d)=0\\
				(a^2 c^2 b^2 d-4)(1+d)=0
				\end{cases}\Leftrightarrow \begin{cases}
				\text{if } d=-1 \text{ then } \begin{cases}
				d=-1\\
				a^2 c^2 b=4
				\end{cases}\\
				\text{if } d=1 \text{ then } \begin{cases}
				d=1\\
				a^2 c^2=4\\
				b^2=1
				\end{cases}\\
				\text{else } \begin{cases}
				b=-1\\
				da^2 c^2=4
				\end{cases}
				\end{cases}$
			\end{center}
			
			The $Z^*$-algebras are then $ Z^{(a,1)} X_{(\frac{2}{a},1)}$, $ Z^{(a,1)} X_{(-\frac{2}{a},1)}$, $ Z^{(a,-1)} X_{(\frac{2}{a},1)}$, $ Z^{(a,-1)} X_{(-\frac{2}{a},1)}$, $Z^{(a,\frac{4}{a^2 b^2})} X_{(b,-1)}$ and $Z^{(a,-1)} X_{(b,\frac{4}{a^2 b^2})}$. The dualizer is the identity for $ Z^{(a,1)} X_{(\frac{2}{a},1)}$.
			
			\item $Z^\alpha H_\beta$: The dualizer of $ZH$ is $\begin{pmatrix}
			2&-1\\-1&1
			\end{pmatrix}$ and its inverse is $\begin{pmatrix}
			1&1\\1&2
			\end{pmatrix}$. let $\alpha=(a,b)$ and $\beta=(c,d)$, $a,b,c,d\in\mathbb{C}^*$. $Z^\alpha$ and $H_\beta$ are compatible iff 
			
			\begin{center}
				$c\begin{pmatrix}
				1&1-d\\0&d
				\end{pmatrix}\begin{pmatrix}
				2&-1\\-1&1
				\end{pmatrix}a\begin{pmatrix}
				1&0\\0&b
				\end{pmatrix}=\frac{1}{a}\begin{pmatrix}
				1&0\\0&\frac{1}{b}
				\end{pmatrix}\begin{pmatrix}
				1&1\\1&2
				\end{pmatrix}\frac{1}{c}\begin{pmatrix}
				1&1-\frac{1}{d}\\0&\frac{1}{d}
				\end{pmatrix}$
			\end{center}
			
			This gives the system:
			
			\begin{center}
				$\begin{cases}
				a^2 c^2(d + 1)=1\\
				a^2 c^2bd=-1\\
				a^2 c^2 b^2 d^2 =1+d
				\end{cases}\Leftrightarrow \begin{cases}
				a^2 c^2\neq 1\\
				b=\frac{1}{a^2 c^2-1}\\
				d=\frac{1-a^2 c^2}{a^2 c^2}
				\end{cases}$
			\end{center}
			
			The $Z^*$-algebras are $ Z^{(a,\frac{1}{a^2 b^2-1})} H_{(b,\frac{1-a^2 b^2}{a^2 b^2})}$ with $a^2 b^2\neq 1$. The dualizer is the Hadamard gate in the case $a=1$ and $b=\sqrt{2}$.
			
			\item $Z^\alpha W_\beta$: the dualizer of $ZW$ is $\begin{pmatrix}
			0&1\\1&0
			\end{pmatrix}$. Let $\alpha=(a,b)$ and $\beta=(c,d)$, $a,b,c\in\mathbb{C}^*$, $d\in\mathbb{C}$. $Z^\alpha$ and $W_\beta$ are compatible iff 
			
			\begin{center}
				$c\begin{pmatrix}
				1&0\\d&1
				\end{pmatrix}\begin{pmatrix}
				0&1\\1&0
				\end{pmatrix}a\begin{pmatrix}
				1&0\\0&b
				\end{pmatrix}=\frac{1}{a}\begin{pmatrix}
				1&0\\0&\frac{1}{b}
				\end{pmatrix}\begin{pmatrix}
				0&1\\1&0
				\end{pmatrix}\frac{1}{c}\begin{pmatrix}
				1&0\\-d&1
				\end{pmatrix}$
			\end{center}
			
			This gives the system: $\begin{cases}
			d=0\\
			a^2 c^2 b=1
			\end{cases}$.
			
			The $Z^*$-algebras are $ Z^{(a,\frac{1}{a^2b^2})} W_{(b,0)}$. The dualizer is the NOT gate in the case $a=1$ and $b=1$.\\
			
			\item $W^\alpha Z_\beta$: The dualizer of $WZ$ is $\begin{pmatrix}
			0&1\\1&0
			\end{pmatrix}$. Let $\alpha=(a,b)$ and $\beta=(c,d)$, $a,c,d\in\mathbb{C}^*$, $b\in\mathbb{C}$. $W^\alpha$ and $Z_\beta$ are compatible iff 
			
			\begin{center}
				$c\begin{pmatrix}
				1&0\\0&d
				\end{pmatrix}\begin{pmatrix}
				0&1\\1&0
				\end{pmatrix}a\begin{pmatrix}
				1&0\\b&1
				\end{pmatrix}=\frac{1}{a}\begin{pmatrix}
				1&0\\-b&1
				\end{pmatrix}\begin{pmatrix}
				0&1\\1&0
				\end{pmatrix}\frac{1}{c}\begin{pmatrix}
				1&0\\0&\frac{1}{d}
				\end{pmatrix}$
			\end{center}
			
			This gives the system: $\begin{cases}
			b=0\\
			c^2a^2 d=1
			\end{cases}$.
			
			The $Z^*$-algebras are $ W^{(a,0)} Z_{(b,\frac{1}{b^2a^2})}$. The dualizer is the NOT gate in the case $a=1$ and $b=1$.\\
			
		\end{itemize}
		
	\end{proof}

	\subsection{Proof of Proposition \ref{monlin}} \phantomsection\label{monlinproof}
	
	\begin{proposition}
		There are only four $Z^*$-algebras in $\textbf{LinRel}_{\mathbb{K}}$: $BB$, $NN$, $BN$ and $NB$.
	\end{proposition}
	
	\begin{proof}
		
		Using the equations of \cite{bonchi2017interacting} this amount to show that $\mu_N$ and $\mu_B$ are the only monoids and that there phase groupos are trivial.
		
		A subspace $M$ of $\mathbb{K}^3$ is unital iff $\exists u\in\mathbb{K}, \forall x, y \in \mathbb{K}, \left((u,x,y)\in M \Leftrightarrow x=y\right) \land \left((x,u,y)\in M \Leftrightarrow x=y\right)$.
		
		The trivial subspaces $\{0,0,0\}$ and $\mathbb{K}^3$ don't satisfy this property.

		If $M$ is of dimension one then there is a vector $(a,b,c)$ such that $\forall x,y,z\in\mathbb{K}, (x,y,z)\in M \Leftrightarrow \exists \lambda\in \mathbb{K}, (x,y,z)=(\lambda a,\lambda b, \lambda c)$.
		
		If $M$ has a unit $u$, given an $x\in\mathbb{K}$ we have $(x,u,x)\in M$ and then $\exists \lambda \mathbb{K}, x=\lambda a, u=\lambda b, x=\lambda c)$. We know that $\lambda\neq 0$ else all triples $(x,u,y)$ would be in $M$. This gives $a=c$, by symmetry we have also $b=c$. The only unital subspace of dimension one is $\mu_N$. It is also associative and thus is a monoid.
		
		If $M$ is of dimension $2$ then there is a vector $(a,b,c)$ such that $\forall x,y,z\in\mathbb{K}, (x,y,z)\in M \Leftrightarrow \exists \lambda\in \mathbb{K}, ax+by+cz=0$.
		
		If $M$ has a unit $u$, given any $x\in\mathbb{K}$ we have $(x,u,x)\in M$ and then $ax+bu+cx=0$. This gives $bu=0$ and $c=-a$. By symmetry we also have $au=0$ and $c=-b$. If $a=b=c=0$ then all triple would be in $M$. We deduce that $a=b=-c\neq 0$ and $u=0$. The only unital subspace of dimension $2$ is $\mu_B$. It is also associative and thus is a monoid.
		
		Finally $\mu_B$ and $\mu_N$ are the only monoids in $\textbf{LinRel}_{\mathbb{K}}$.
		 
		Now let $\alpha$ be a phase of $\mu_N$, if $(x,y)\in \alpha$ then the phase's definition gives us that $x=y$, so $\alpha=id$ or $\alpha= {(0,0)}$ the only invertible possibility is $id$. Now let $\beta$ be a phase of $\mu_B$, if $(x,y)\in \alpha$ then the phase's definition gives us that for all $z\in\mathbb{K}$ $(x+z,y+z)\in\alpha$, thus $\alpha=id$ or $\alpha=\mathbb{K}^2$ the only invertible possibility is $id$. Finally both phase groups are trivial.
	\end{proof}
	

\section{All graphical calculi for quantum computing}

\subsection[The ZZ-calculi]{The $ZZ$-calculis}

\subsubsection[First ZZ calculus]{$Z^{(a,b/a)}Z_{(1/a,a/b)}$}

This is the first calculus presented in the Theorem, up to a
re-parametrization that makes it slightly better looking:

\input figs/mprodwhite $=\begin{pmatrix}
a & 0 & 0 & 0 \\
0 & 0 & 0 & b
\end{pmatrix}$ \hfill
\begin{tikzpicture}[baseline=(current bounding box.center)]
\begin{pgfonlayer}{nodelayer}
\coordinate (0) at (0.25,0);
\node [style=gmonoid] (1) at (0.25,0.5) {};

\end{pgfonlayer}
\begin{pgfonlayer}{edgelayer}
\draw[] (0.center) to[out= 90, in=-89.9] (1.center);

\end{pgfonlayer}
\end{tikzpicture}
 $=\begin{pmatrix}
\frac{1}{a} \\
\frac{1}{b}
\end{pmatrix}$ \hfill
\input figs/coprodwhite $=\begin{pmatrix}
1 & 0 \\
0 & 0 \\
0 & 0 \\
0 & 1
\end{pmatrix}$ \hfill
\begin{tikzpicture}[baseline=(current bounding box.center)]
\begin{pgfonlayer}{nodelayer}
\node [style=gmonoid] (0) at (0.25,0.25) {};
\coordinate (1) at (0.25,0.75);

\end{pgfonlayer}
\begin{pgfonlayer}{edgelayer}
\draw[] (0.center) to[out= 90, in=-89.9] (1.center);

\end{pgfonlayer}
\end{tikzpicture}
 $=\begin{pmatrix}
1 & 1
\end{pmatrix}$\hfill
\input figs/phasewhite  $=x\begin{pmatrix}
1 & 0 \\
0 & y
\end{pmatrix}$

\hfill\begin{tikzpicture}[baseline=(current bounding box.center)]
\begin{pgfonlayer}{nodelayer}
\node [] (0) at (0.5,0) {};
\coordinate (1) at (0.25,0.5);
\coordinate (2) at (0.75,0.5);

\end{pgfonlayer}
\begin{pgfonlayer}{edgelayer}
\draw[] (0.center) to[out=180, in=-89.9] (1.center);
\draw[] (0.center) to[out=  0, in=-89.9] (2.center);

\end{pgfonlayer}
\end{tikzpicture}
 $=\begin{pmatrix}
a & 0 & 0 & b
\end{pmatrix}$\hfill
\begin{tikzpicture}[baseline=(current bounding box.center)]
\begin{pgfonlayer}{nodelayer}
\coordinate (0) at (0.25,0);
\coordinate (1) at (0.75,0);
\node [] (2) at (0.5,0.5) {};

\end{pgfonlayer}
\begin{pgfonlayer}{edgelayer}
\draw[] (0.center) to[out= 90, in=180] (2.center);
\draw[] (1.center) to[out= 90, in=  0] (2.center);

\end{pgfonlayer}
\end{tikzpicture}
 $=\begin{pmatrix}
\frac{1}{a} \\
0 \\
0 \\
\frac{1}{b}
\end{pmatrix}$\hfill
\begin{tikzpicture}[baseline=(current bounding box.center)]
\begin{pgfonlayer}{nodelayer}
\coordinate (0) at (0.25,0);
\node [style=hmonoid] (1) at (0.25,0.5) {};
\coordinate (2) at (0.25,1);

\end{pgfonlayer}
\begin{pgfonlayer}{edgelayer}
\draw[] (0.center) to[out= 90, in=-89.9] (1.center);
\draw[] (1.center) to[out= 90, in=-89.9] (2.center);

\end{pgfonlayer}
\end{tikzpicture}
 $=\begin{pmatrix}
1 & 0 \\
0 & 1
\end{pmatrix}$\hfill

\input figs/mprodblack $=\begin{pmatrix}
1 & 0 & 0 & 0 \\
0 & 0 & 0 & 1
\end{pmatrix}$ \hfill
\begin{tikzpicture}[baseline=(current bounding box.center)]
\begin{pgfonlayer}{nodelayer}
\coordinate (0) at (0.25,0);
\node [style=hmonoid] (1) at (0.25,0.5) {};

\end{pgfonlayer}
\begin{pgfonlayer}{edgelayer}
\draw[] (0.center) to[out= 90, in=-89.9] (1.center);

\end{pgfonlayer}
\end{tikzpicture}
 $=\begin{pmatrix}
1 \\
1
\end{pmatrix}$ \hfill
\input figs/coprodblack $=\begin{pmatrix}
\frac{1}{a} & 0 \\
0 & 0 \\
0 & 0 \\
0 & \frac{1}{b}
\end{pmatrix}$ \hfill
\begin{tikzpicture}[baseline=(current bounding box.center)]
\begin{pgfonlayer}{nodelayer}
\node [style=hmonoid] (0) at (0.25,0.25) {};
\coordinate (1) at (0.25,0.75);

\end{pgfonlayer}
\begin{pgfonlayer}{edgelayer}
\draw[] (0.center) to[out= 90, in=-89.9] (1.center);

\end{pgfonlayer}
\end{tikzpicture}
 $=\begin{pmatrix}
a & b
\end{pmatrix}$\hfill
\input figs/phaseblack $=x\begin{pmatrix}
1 & 0 \\
0 & y
\end{pmatrix}$

\subsubsection[Second ZZ calculus]{$Z^{(a,b/a)}Z_{(1/a,-a/b)}$}

The only difference with the previous calculus is in the following generators:

\begin{tikzpicture}[baseline=(current bounding box.center)]
\begin{pgfonlayer}{nodelayer}
\coordinate (0) at (0.25,0);
\node [style=hmonoid] (1) at (0.25,0.5) {};
\coordinate (2) at (0.25,1);

\end{pgfonlayer}
\begin{pgfonlayer}{edgelayer}
\draw[] (0.center) to[out= 90, in=-89.9] (1.center);
\draw[] (1.center) to[out= 90, in=-89.9] (2.center);

\end{pgfonlayer}
\end{tikzpicture}
 $=\begin{pmatrix}
1 & 0 \\
0 & -1
\end{pmatrix}$\hfill
\input figs/mprodblack $=\begin{pmatrix}
1 & 0 & 0 & 0 \\
0 & 0 & 0 & -1
\end{pmatrix}$ \hfill
\begin{tikzpicture}[baseline=(current bounding box.center)]
\begin{pgfonlayer}{nodelayer}
\coordinate (0) at (0.25,0);
\node [style=hmonoid] (1) at (0.25,0.5) {};

\end{pgfonlayer}
\begin{pgfonlayer}{edgelayer}
\draw[] (0.center) to[out= 90, in=-89.9] (1.center);

\end{pgfonlayer}
\end{tikzpicture}
 $=\begin{pmatrix}
1 \\
-1
\end{pmatrix}$

\input figs/coprodblack $=\begin{pmatrix}
\frac{1}{a} & 0 \\
0 & 0 \\
0 & 0 \\
0 & -\frac{1}{b}
\end{pmatrix}$ \hfill
\begin{tikzpicture}[baseline=(current bounding box.center)]
\begin{pgfonlayer}{nodelayer}
\node [style=hmonoid] (0) at (0.25,0.25) {};
\coordinate (1) at (0.25,0.75);

\end{pgfonlayer}
\begin{pgfonlayer}{edgelayer}
\draw[] (0.center) to[out= 90, in=-89.9] (1.center);

\end{pgfonlayer}
\end{tikzpicture}
 $=\begin{pmatrix}
a & -b
\end{pmatrix}$\hfill
\input figs/phaseblack $=x\begin{pmatrix}
1 & 0 \\
0 & -y
\end{pmatrix}$

\subsubsection[The other ZZ-calculi]{$Z^{(a,b/a)}Z_{(-1/a,a/b)}$ and $Z^{(a,b/a)}Z_{(-1/a,-a/b)}$}

These calculi differ from the previous ones only by the presence of a global
scalar ``-1'' in all matrices corresponding to the black nodes.

\subsection{The ZX-calculi}

\subsubsection[The first ZX calculus]{$Z^{(a,1)}X_{(2/a, 1)}$}

In the case $a = 1$, this is almost the $ZX$-calculus of \cite{coecke2011interacting}:

\input figs/mprodwhite $=a\begin{pmatrix}
1 & 0 & 0 & 0 \\
0 & 0 & 0 & 1
\end{pmatrix}$ \hfill
\begin{tikzpicture}[baseline=(current bounding box.center)]
\begin{pgfonlayer}{nodelayer}
\coordinate (0) at (0.25,0);
\node [style=gmonoid] (1) at (0.25,0.5) {};

\end{pgfonlayer}
\begin{pgfonlayer}{edgelayer}
\draw[] (0.center) to[out= 90, in=-89.9] (1.center);

\end{pgfonlayer}
\end{tikzpicture}
 $=\frac{1}{a}\begin{pmatrix}
1 \\
1
\end{pmatrix}$ \hfill
\input figs/coprodwhite $=\begin{pmatrix}
1 & 0 \\
0 & 0 \\
0 & 0 \\
0 & 1
\end{pmatrix}$ \hfill
\begin{tikzpicture}[baseline=(current bounding box.center)]
\begin{pgfonlayer}{nodelayer}
\node [style=gmonoid] (0) at (0.25,0.25) {};
\coordinate (1) at (0.25,0.75);

\end{pgfonlayer}
\begin{pgfonlayer}{edgelayer}
\draw[] (0.center) to[out= 90, in=-89.9] (1.center);

\end{pgfonlayer}
\end{tikzpicture}
 $=\begin{pmatrix}
1 & 1
\end{pmatrix}$\hfill
\input figs/phasewhite  $=x\begin{pmatrix}
1 & 0 \\
0 & y
\end{pmatrix}$

\hfill\begin{tikzpicture}[baseline=(current bounding box.center)]
\begin{pgfonlayer}{nodelayer}
\node [] (0) at (0.5,0) {};
\coordinate (1) at (0.25,0.5);
\coordinate (2) at (0.75,0.5);

\end{pgfonlayer}
\begin{pgfonlayer}{edgelayer}
\draw[] (0.center) to[out=180, in=-89.9] (1.center);
\draw[] (0.center) to[out=  0, in=-89.9] (2.center);

\end{pgfonlayer}
\end{tikzpicture}
 $= a\begin{pmatrix}
1 & 0 & 0 & 1
\end{pmatrix}$\hfill
\begin{tikzpicture}[baseline=(current bounding box.center)]
\begin{pgfonlayer}{nodelayer}
\coordinate (0) at (0.25,0);
\coordinate (1) at (0.75,0);
\node [] (2) at (0.5,0.5) {};

\end{pgfonlayer}
\begin{pgfonlayer}{edgelayer}
\draw[] (0.center) to[out= 90, in=180] (2.center);
\draw[] (1.center) to[out= 90, in=  0] (2.center);

\end{pgfonlayer}
\end{tikzpicture}
 $=\frac{1}{a}\begin{pmatrix}
1 \\
0 \\
0 \\
1
\end{pmatrix}$\hfill
\begin{tikzpicture}[baseline=(current bounding box.center)]
\begin{pgfonlayer}{nodelayer}
\coordinate (0) at (0.25,0);
\node [style=hmonoid] (1) at (0.25,0.5) {};
\coordinate (2) at (0.25,1);

\end{pgfonlayer}
\begin{pgfonlayer}{edgelayer}
\draw[] (0.center) to[out= 90, in=-89.9] (1.center);
\draw[] (1.center) to[out= 90, in=-89.9] (2.center);

\end{pgfonlayer}
\end{tikzpicture}
 $=\begin{pmatrix}
1 & 0 \\
0 & 1
\end{pmatrix}$\hfill

\input figs/mprodblack $=\begin{pmatrix}
1 & 0 & 0 & 1 \\
0 & 1 & 1 & 0
\end{pmatrix}$ \hfill
\begin{tikzpicture}[baseline=(current bounding box.center)]
\begin{pgfonlayer}{nodelayer}
\coordinate (0) at (0.25,0);
\node [style=hmonoid] (1) at (0.25,0.5) {};

\end{pgfonlayer}
\begin{pgfonlayer}{edgelayer}
\draw[] (0.center) to[out= 90, in=-89.9] (1.center);

\end{pgfonlayer}
\end{tikzpicture}
 $=\begin{pmatrix}
1 \\
0
\end{pmatrix}$ \hfill
\input figs/coprodblack $=\frac{1}{a}\begin{pmatrix}
1 & 0 \\
0 & 1 \\
0 & 1 \\
1 & 0
\end{pmatrix}$ \hfill
\begin{tikzpicture}[baseline=(current bounding box.center)]
\begin{pgfonlayer}{nodelayer}
\node [style=hmonoid] (0) at (0.25,0.25) {};
\coordinate (1) at (0.25,0.75);

\end{pgfonlayer}
\begin{pgfonlayer}{edgelayer}
\draw[] (0.center) to[out= 90, in=-89.9] (1.center);

\end{pgfonlayer}
\end{tikzpicture}
 $=a\begin{pmatrix}
1 & 0
\end{pmatrix}$

\input figs/phaseblack $=\frac{1}{2} \, x\begin{pmatrix}
y + 1 & -y + 1 \\
-y + 1 & y + 1
\end{pmatrix}$

\subsubsection[The second ZX calculus]{$Z^{(a,1)}X_{(2/a, -1)}$}

The only difference with the previous calculus is in the following generators:

\begin{tikzpicture}[baseline=(current bounding box.center)]
\begin{pgfonlayer}{nodelayer}
\coordinate (0) at (0.25,0);
\node [style=hmonoid] (1) at (0.25,0.5) {};
\coordinate (2) at (0.25,1);

\end{pgfonlayer}
\begin{pgfonlayer}{edgelayer}
\draw[] (0.center) to[out= 90, in=-89.9] (1.center);
\draw[] (1.center) to[out= 90, in=-89.9] (2.center);

\end{pgfonlayer}
\end{tikzpicture}
 $=\begin{pmatrix}
0 & 1 \\
1 & 0
\end{pmatrix}$
\hfill
\input figs/mprodblack $=\begin{pmatrix}
0 & 1 & 1 & 0 \\
1 & 0 & 0 & 1
\end{pmatrix}$ \hfill
\begin{tikzpicture}[baseline=(current bounding box.center)]
\begin{pgfonlayer}{nodelayer}
\coordinate (0) at (0.25,0);
\node [style=hmonoid] (1) at (0.25,0.5) {};

\end{pgfonlayer}
\begin{pgfonlayer}{edgelayer}
\draw[] (0.center) to[out= 90, in=-89.9] (1.center);

\end{pgfonlayer}
\end{tikzpicture}
 $=\begin{pmatrix}
0 \\
1
\end{pmatrix}$

\input figs/coprodblack $=\frac{1}{a}\begin{pmatrix}
0 & 1 \\
1 & 0 \\
1 & 0 \\
0 & 1
\end{pmatrix}$ 
\hfill
\begin{tikzpicture}[baseline=(current bounding box.center)]
\begin{pgfonlayer}{nodelayer}
\node [style=hmonoid] (0) at (0.25,0.25) {};
\coordinate (1) at (0.25,0.75);

\end{pgfonlayer}
\begin{pgfonlayer}{edgelayer}
\draw[] (0.center) to[out= 90, in=-89.9] (1.center);

\end{pgfonlayer}
\end{tikzpicture}
 $=a\begin{pmatrix}
0 & 1
\end{pmatrix}$\hfill
\input figs/phaseblack $=\frac{1}{2} \, x\begin{pmatrix}
-y + 1 & y + 1 \\
y + 1 & -y + 1
\end{pmatrix}$

\subsubsection[The third and fourth ZX calculi]{$Z^{(a,1)}X_{(-2/a, -1)}$ and $Z^{(a,1)}X_{(-2/a, 1)}$}

These calculi differ from the previous ones only by the presence of a global
scalar ``-1'' in all matrices corresponding to the black nodes.

\clearpage
\subsubsection[The fifth ZX calculus]{$Z^{(a/b,b^2)}X_{(2/a,-1)}$}

This is a quite different calculus:

\input figs/mprodwhite $=a\begin{pmatrix}
\frac{1}{b} & 0 & 0 & 0 \\
0 & 0 & 0 & b
\end{pmatrix}$ \hfill
\begin{tikzpicture}[baseline=(current bounding box.center)]
\begin{pgfonlayer}{nodelayer}
\coordinate (0) at (0.25,0);
\node [style=gmonoid] (1) at (0.25,0.5) {};

\end{pgfonlayer}
\begin{pgfonlayer}{edgelayer}
\draw[] (0.center) to[out= 90, in=-89.9] (1.center);

\end{pgfonlayer}
\end{tikzpicture}
 $=\frac{1}{a}\begin{pmatrix}
b \\
\frac{1}{b}
\end{pmatrix}$ \hfill
\input figs/coprodwhite $=\begin{pmatrix}
1 & 0 \\
0 & 0 \\
0 & 0 \\
0 & 1
\end{pmatrix}$ \hfill
\begin{tikzpicture}[baseline=(current bounding box.center)]
\begin{pgfonlayer}{nodelayer}
\node [style=gmonoid] (0) at (0.25,0.25) {};
\coordinate (1) at (0.25,0.75);

\end{pgfonlayer}
\begin{pgfonlayer}{edgelayer}
\draw[] (0.center) to[out= 90, in=-89.9] (1.center);

\end{pgfonlayer}
\end{tikzpicture}
 $=\begin{pmatrix}
1 & 1
\end{pmatrix}$\hfill
\input figs/phasewhite  $=x\begin{pmatrix}
1 & 0 \\
0 & y
\end{pmatrix}$

\hfill\begin{tikzpicture}[baseline=(current bounding box.center)]
\begin{pgfonlayer}{nodelayer}
\node [] (0) at (0.5,0) {};
\coordinate (1) at (0.25,0.5);
\coordinate (2) at (0.75,0.5);

\end{pgfonlayer}
\begin{pgfonlayer}{edgelayer}
\draw[] (0.center) to[out=180, in=-89.9] (1.center);
\draw[] (0.center) to[out=  0, in=-89.9] (2.center);

\end{pgfonlayer}
\end{tikzpicture}
 $= a\begin{pmatrix}
\frac{1}{b} & 0 & 0 & b
\end{pmatrix}$\hfill
\begin{tikzpicture}[baseline=(current bounding box.center)]
\begin{pgfonlayer}{nodelayer}
\coordinate (0) at (0.25,0);
\coordinate (1) at (0.75,0);
\node [] (2) at (0.5,0.5) {};

\end{pgfonlayer}
\begin{pgfonlayer}{edgelayer}
\draw[] (0.center) to[out= 90, in=180] (2.center);
\draw[] (1.center) to[out= 90, in=  0] (2.center);

\end{pgfonlayer}
\end{tikzpicture}
 $=\frac{1}{a}\begin{pmatrix}
b \\
0 \\
0 \\
\frac{1}{b}
\end{pmatrix}$\hfill
\begin{tikzpicture}[baseline=(current bounding box.center)]
\begin{pgfonlayer}{nodelayer}
\coordinate (0) at (0.25,0);
\node [style=hmonoid] (1) at (0.25,0.5) {};
\coordinate (2) at (0.25,1);

\end{pgfonlayer}
\begin{pgfonlayer}{edgelayer}
\draw[] (0.center) to[out= 90, in=-89.9] (1.center);
\draw[] (1.center) to[out= 90, in=-89.9] (2.center);

\end{pgfonlayer}
\end{tikzpicture}
 $=\begin{pmatrix}
0 & b \\
\frac{1}{b} & 0
\end{pmatrix}$\hfill

\input figs/mprodblack $=\begin{pmatrix}
0 & b & b & 0 \\
\frac{1}{b} & 0 & 0 & \frac{1}{b}
\end{pmatrix}$ \hfill
\begin{tikzpicture}[baseline=(current bounding box.center)]
\begin{pgfonlayer}{nodelayer}
\coordinate (0) at (0.25,0);
\node [style=hmonoid] (1) at (0.25,0.5) {};

\end{pgfonlayer}
\begin{pgfonlayer}{edgelayer}
\draw[] (0.center) to[out= 90, in=-89.9] (1.center);

\end{pgfonlayer}
\end{tikzpicture}
 $=\begin{pmatrix}
0 \\
\frac{1}{b}
\end{pmatrix}$ \hfill
\input figs/coprodblack $=\frac{1}{a}\begin{pmatrix}
0 & b^{2} \\
1 & 0 \\
1 & 0 \\
0 & \frac{1}{b^{2}}
\end{pmatrix}$ \hfill
\begin{tikzpicture}[baseline=(current bounding box.center)]
\begin{pgfonlayer}{nodelayer}
\node [style=hmonoid] (0) at (0.25,0.25) {};
\coordinate (1) at (0.25,0.75);

\end{pgfonlayer}
\begin{pgfonlayer}{edgelayer}
\draw[] (0.center) to[out= 90, in=-89.9] (1.center);

\end{pgfonlayer}
\end{tikzpicture}
 $=a\begin{pmatrix}
0 & 1
\end{pmatrix}$

\input figs/phaseblack $=\frac{1}{2} \, x\begin{pmatrix}
-b y + b & b y + b \\
\frac{y + 1}{b} & -\frac{y - 1}{b}
\end{pmatrix}$
\subsubsection[The last ZX calculus]{$Z^{(a,-1)}X_{(2b/a,1/b^2)}$}

This is a calculus dual to the previous one, but the equations look
more intricate:

\input figs/mprodwhite $=a\begin{pmatrix}
1 & 0 & 0 & 0 \\
0 & 0 & 0 & -1
\end{pmatrix}$ \hfill
\begin{tikzpicture}[baseline=(current bounding box.center)]
\begin{pgfonlayer}{nodelayer}
\coordinate (0) at (0.25,0);
\node [style=gmonoid] (1) at (0.25,0.5) {};

\end{pgfonlayer}
\begin{pgfonlayer}{edgelayer}
\draw[] (0.center) to[out= 90, in=-89.9] (1.center);

\end{pgfonlayer}
\end{tikzpicture}
 $=\frac{1}{a}\begin{pmatrix}
1 \\
-1
\end{pmatrix}$ \hfill
\input figs/coprodwhite $=\begin{pmatrix}
1 & 0 \\
0 & 0 \\
0 & 0 \\
0 & 1
\end{pmatrix}$ \hfill
\begin{tikzpicture}[baseline=(current bounding box.center)]
\begin{pgfonlayer}{nodelayer}
\node [style=gmonoid] (0) at (0.25,0.25) {};
\coordinate (1) at (0.25,0.75);

\end{pgfonlayer}
\begin{pgfonlayer}{edgelayer}
\draw[] (0.center) to[out= 90, in=-89.9] (1.center);

\end{pgfonlayer}
\end{tikzpicture}
 $=\begin{pmatrix}
1 & 1
\end{pmatrix}$\hfill
\input figs/phasewhite  $=x\begin{pmatrix}
1 & 0 \\
0 & y
\end{pmatrix}$

\hfill\begin{tikzpicture}[baseline=(current bounding box.center)]
\begin{pgfonlayer}{nodelayer}
\node [] (0) at (0.5,0) {};
\coordinate (1) at (0.25,0.5);
\coordinate (2) at (0.75,0.5);

\end{pgfonlayer}
\begin{pgfonlayer}{edgelayer}
\draw[] (0.center) to[out=180, in=-89.9] (1.center);
\draw[] (0.center) to[out=  0, in=-89.9] (2.center);

\end{pgfonlayer}
\end{tikzpicture}
 $= a\begin{pmatrix}
1 & 0 & 0 & -1
\end{pmatrix}$\hfill
\begin{tikzpicture}[baseline=(current bounding box.center)]
\begin{pgfonlayer}{nodelayer}
\coordinate (0) at (0.25,0);
\coordinate (1) at (0.75,0);
\node [] (2) at (0.5,0.5) {};

\end{pgfonlayer}
\begin{pgfonlayer}{edgelayer}
\draw[] (0.center) to[out= 90, in=180] (2.center);
\draw[] (1.center) to[out= 90, in=  0] (2.center);

\end{pgfonlayer}
\end{tikzpicture}
 $=\frac{1}{a}\begin{pmatrix}
1 \\
0 \\
0 \\
-1
\end{pmatrix}$\hfill
\begin{tikzpicture}[baseline=(current bounding box.center)]
\begin{pgfonlayer}{nodelayer}
\coordinate (0) at (0.25,0);
\node [style=hmonoid] (1) at (0.25,0.5) {};
\coordinate (2) at (0.25,1);

\end{pgfonlayer}
\begin{pgfonlayer}{edgelayer}
\draw[] (0.center) to[out= 90, in=-89.9] (1.center);
\draw[] (1.center) to[out= 90, in=-89.9] (2.center);

\end{pgfonlayer}
\end{tikzpicture}
 $=\frac{1}{2b}\begin{pmatrix}
b^{2} + 1 & 1 - b^2 \\
b^{2} - 1 & -1-b^2\\
\end{pmatrix}$\hfill

\input figs/mprodblack $=\frac{1}{2b}\begin{pmatrix}
b^{2} + 1 & 1 - b^2 & 1 - b^2 & b^{2} + 1 \\
b^{2} - 1 & -1 - b^2 & - 1 - b^2 & b^{2} - 1
\end{pmatrix}$ \hfill
\begin{tikzpicture}[baseline=(current bounding box.center)]
\begin{pgfonlayer}{nodelayer}
\coordinate (0) at (0.25,0);
\node [style=hmonoid] (1) at (0.25,0.5) {};

\end{pgfonlayer}
\begin{pgfonlayer}{edgelayer}
\draw[] (0.center) to[out= 90, in=-89.9] (1.center);

\end{pgfonlayer}
\end{tikzpicture}
 $=\frac{1}{2b}\begin{pmatrix}
b^{2} + 1 \\
b^{2} - 1
\end{pmatrix}$

\input figs/coprodblack $=\frac{1}{2ab}\begin{pmatrix}
b^{2} + 1 & 1 - b^2 \\
b^{2} - 1 & - 1 - b^2 \\
b^{2} - 1 & - 1 - b^2 \\
b^{2} + 1 & 1 - b^2\\
\end{pmatrix}$
\hfill
\begin{tikzpicture}[baseline=(current bounding box.center)]
\begin{pgfonlayer}{nodelayer}
\node [style=hmonoid] (0) at (0.25,0.25) {};
\coordinate (1) at (0.25,0.75);

\end{pgfonlayer}
\begin{pgfonlayer}{edgelayer}
\draw[] (0.center) to[out= 90, in=-89.9] (1.center);

\end{pgfonlayer}
\end{tikzpicture}
 $=\frac{a}{2b}\begin{pmatrix}
b^{2} + 1 & 1 - b^2
\end{pmatrix}$

\input figs/phaseblack $=\frac{x}{2b}\begin{pmatrix}
b^{2} y + 1 & 1 - b^{2} y \\
b^{2} y - 1 & -1 - b^2 y\\
\end{pmatrix}$
\clearpage
\subsection{The ZH-calculi}

\subsubsection[The First ZH-calculus]{The $Z^{(a,1/(b^2-1))}H_{(b/a,(1-b^2)/b^2)}$ calculus}

\input figs/mprodwhite $=a\begin{pmatrix}
1 & 0 & 0 & 0 \\
0 & 0 & 0 & \frac{1}{b^{2} - 1}
\end{pmatrix}$ \hfill
\begin{tikzpicture}[baseline=(current bounding box.center)]
\begin{pgfonlayer}{nodelayer}
\coordinate (0) at (0.25,0);
\node [style=gmonoid] (1) at (0.25,0.5) {};

\end{pgfonlayer}
\begin{pgfonlayer}{edgelayer}
\draw[] (0.center) to[out= 90, in=-89.9] (1.center);

\end{pgfonlayer}
\end{tikzpicture}
 $=\frac{1}{a}\begin{pmatrix}
1 \\
b^{2} - 1
\end{pmatrix}$ \hfill
\input figs/coprodwhite $=\begin{pmatrix}
1 & 0 \\
0 & 0 \\
0 & 0 \\
0 & 1
\end{pmatrix}$ \hfill
\begin{tikzpicture}[baseline=(current bounding box.center)]
\begin{pgfonlayer}{nodelayer}
\node [style=gmonoid] (0) at (0.25,0.25) {};
\coordinate (1) at (0.25,0.75);

\end{pgfonlayer}
\begin{pgfonlayer}{edgelayer}
\draw[] (0.center) to[out= 90, in=-89.9] (1.center);

\end{pgfonlayer}
\end{tikzpicture}
 $=\begin{pmatrix}
1 & 1
\end{pmatrix}$\hfill
\input figs/phasewhite  $=x\begin{pmatrix}
1 & 0 \\
0 & y
\end{pmatrix}$

\hfill\begin{tikzpicture}[baseline=(current bounding box.center)]
\begin{pgfonlayer}{nodelayer}
\node [] (0) at (0.5,0) {};
\coordinate (1) at (0.25,0.5);
\coordinate (2) at (0.75,0.5);

\end{pgfonlayer}
\begin{pgfonlayer}{edgelayer}
\draw[] (0.center) to[out=180, in=-89.9] (1.center);
\draw[] (0.center) to[out=  0, in=-89.9] (2.center);

\end{pgfonlayer}
\end{tikzpicture}
 $= a\begin{pmatrix}
1 & 0 & 0 & \frac{1}{b^{2} - 1}
\end{pmatrix}$\hfill
\begin{tikzpicture}[baseline=(current bounding box.center)]
\begin{pgfonlayer}{nodelayer}
\coordinate (0) at (0.25,0);
\coordinate (1) at (0.75,0);
\node [] (2) at (0.5,0.5) {};

\end{pgfonlayer}
\begin{pgfonlayer}{edgelayer}
\draw[] (0.center) to[out= 90, in=180] (2.center);
\draw[] (1.center) to[out= 90, in=  0] (2.center);

\end{pgfonlayer}
\end{tikzpicture}
 $=\frac{1}{a}\begin{pmatrix}
1 \\
0 \\
0 \\
b^{2} - 1
\end{pmatrix}$\hfill
\begin{tikzpicture}[baseline=(current bounding box.center)]
\begin{pgfonlayer}{nodelayer}
\coordinate (0) at (0.25,0);
\node [style=hmonoid] (1) at (0.25,0.5) {};
\coordinate (2) at (0.25,1);

\end{pgfonlayer}
\begin{pgfonlayer}{edgelayer}
\draw[] (0.center) to[out= 90, in=-89.9] (1.center);
\draw[] (1.center) to[out= 90, in=-89.9] (2.center);

\end{pgfonlayer}
\end{tikzpicture}
 $=\frac{1}{b}\begin{pmatrix}
1 & 1 \\
b^{2} - 1 & -1
\end{pmatrix}$\hfill

\input figs/mprodblack $=\frac{1}{b}\begin{pmatrix}
1 & 1 & 1 & 1 \\
b^{2} - 1 & b^{2} - 1 & b^{2} - 1& -1
\end{pmatrix}$ \hfill
\begin{tikzpicture}[baseline=(current bounding box.center)]
\begin{pgfonlayer}{nodelayer}
\coordinate (0) at (0.25,0);
\node [style=hmonoid] (1) at (0.25,0.5) {};

\end{pgfonlayer}
\begin{pgfonlayer}{edgelayer}
\draw[] (0.center) to[out= 90, in=-89.9] (1.center);

\end{pgfonlayer}
\end{tikzpicture}
 $=\frac{1}{b}\begin{pmatrix}
1\\
-1
\end{pmatrix}$ \hfill
\input figs/coprodblack $=\frac{1}{ab}\begin{pmatrix}
1 & 1 \\
b^{2} - 1 & b^{2} - 1\\
b^{2} - 1 & b^{2} - 1\\
(b^{2} - 1)^2 & 1 - b^2\\
\end{pmatrix}$

\begin{tikzpicture}[baseline=(current bounding box.center)]
\begin{pgfonlayer}{nodelayer}
\node [style=hmonoid] (0) at (0.25,0.25) {};
\coordinate (1) at (0.25,0.75);

\end{pgfonlayer}
\begin{pgfonlayer}{edgelayer}
\draw[] (0.center) to[out= 90, in=-89.9] (1.center);

\end{pgfonlayer}
\end{tikzpicture}
 $=\frac{a}{b}\begin{pmatrix}
1 & \frac{1}{1 - b^{2}}
\end{pmatrix}$\hfill
\input figs/phaseblack $=\frac{x}{b}\begin{pmatrix}
1 & 1 \\
b^{2} - 1 & b^2 - 1 - b^2y\\
\end{pmatrix}$

\subsubsection[The First ZH-calculus revisited]{The $Z^{(a,1/c}H_{(\sqrt{c+1}/a,-c/(c+1))}$ calculus}

This is an alternative presentation of the previous calculus taking $c = b^2 - 1$. While not a new calculus \textit{per se}, we think it is easier to
understand than the previous one. $\sqrt{c+1}$ represents one of the two square roots of $c+1$.

\input figs/mprodwhite $=a\begin{pmatrix}
1 & 0 & 0 & 0 \\
0 & 0 & 0 & \frac{1}{c}
\end{pmatrix}$ \hfill
\begin{tikzpicture}[baseline=(current bounding box.center)]
\begin{pgfonlayer}{nodelayer}
\coordinate (0) at (0.25,0);
\node [style=gmonoid] (1) at (0.25,0.5) {};

\end{pgfonlayer}
\begin{pgfonlayer}{edgelayer}
\draw[] (0.center) to[out= 90, in=-89.9] (1.center);

\end{pgfonlayer}
\end{tikzpicture}
 $=\frac{1}{a}\begin{pmatrix}
1 \\
c
\end{pmatrix}$ \hfill
\input figs/coprodwhite $=\begin{pmatrix}
1 & 0 \\
0 & 0 \\
0 & 0 \\
0 & 1
\end{pmatrix}$ \hfill
\begin{tikzpicture}[baseline=(current bounding box.center)]
\begin{pgfonlayer}{nodelayer}
\node [style=gmonoid] (0) at (0.25,0.25) {};
\coordinate (1) at (0.25,0.75);

\end{pgfonlayer}
\begin{pgfonlayer}{edgelayer}
\draw[] (0.center) to[out= 90, in=-89.9] (1.center);

\end{pgfonlayer}
\end{tikzpicture}
 $=\begin{pmatrix}
1 & 1
\end{pmatrix}$\hfill
\input figs/phasewhite  $=x\begin{pmatrix}
1 & 0 \\
0 & y
\end{pmatrix}$

\hfill\begin{tikzpicture}[baseline=(current bounding box.center)]
\begin{pgfonlayer}{nodelayer}
\node [] (0) at (0.5,0) {};
\coordinate (1) at (0.25,0.5);
\coordinate (2) at (0.75,0.5);

\end{pgfonlayer}
\begin{pgfonlayer}{edgelayer}
\draw[] (0.center) to[out=180, in=-89.9] (1.center);
\draw[] (0.center) to[out=  0, in=-89.9] (2.center);

\end{pgfonlayer}
\end{tikzpicture}
 $= a\begin{pmatrix}
1 & 0 & 0 & \frac{1}{c}
\end{pmatrix}$\hfill
\begin{tikzpicture}[baseline=(current bounding box.center)]
\begin{pgfonlayer}{nodelayer}
\coordinate (0) at (0.25,0);
\coordinate (1) at (0.75,0);
\node [] (2) at (0.5,0.5) {};

\end{pgfonlayer}
\begin{pgfonlayer}{edgelayer}
\draw[] (0.center) to[out= 90, in=180] (2.center);
\draw[] (1.center) to[out= 90, in=  0] (2.center);

\end{pgfonlayer}
\end{tikzpicture}
 $=\frac{1}{a}\begin{pmatrix}
1 \\
0 \\
0 \\
c
\end{pmatrix}$\hfill
\begin{tikzpicture}[baseline=(current bounding box.center)]
\begin{pgfonlayer}{nodelayer}
\coordinate (0) at (0.25,0);
\node [style=hmonoid] (1) at (0.25,0.5) {};
\coordinate (2) at (0.25,1);

\end{pgfonlayer}
\begin{pgfonlayer}{edgelayer}
\draw[] (0.center) to[out= 90, in=-89.9] (1.center);
\draw[] (1.center) to[out= 90, in=-89.9] (2.center);

\end{pgfonlayer}
\end{tikzpicture}
 $=\frac{1}{\sqrt{c+1}}\begin{pmatrix}
1 & 1 \\
c & - 1\\
\end{pmatrix}$\hfill

\input figs/mprodblack $=\frac{1}{\sqrt{c+1}}\begin{pmatrix}
1 & 1 & 1 & 1\\
c & c & c & -1\\
\end{pmatrix}$ \hfill
\begin{tikzpicture}[baseline=(current bounding box.center)]
\begin{pgfonlayer}{nodelayer}
\coordinate (0) at (0.25,0);
\node [style=hmonoid] (1) at (0.25,0.5) {};

\end{pgfonlayer}
\begin{pgfonlayer}{edgelayer}
\draw[] (0.center) to[out= 90, in=-89.9] (1.center);

\end{pgfonlayer}
\end{tikzpicture}
 $=\frac{1}{\sqrt{c+1}}\begin{pmatrix}
1 \\
-1
\end{pmatrix}$ \hfill
\input figs/coprodblack $=\frac{1}{a\sqrt{c+1}}\begin{pmatrix}
1 & 1 \\
c & c \\
c & c \\
c^{2} & -c\\
\end{pmatrix}$ \hfill
\begin{tikzpicture}[baseline=(current bounding box.center)]
\begin{pgfonlayer}{nodelayer}
\node [style=hmonoid] (0) at (0.25,0.25) {};
\coordinate (1) at (0.25,0.75);

\end{pgfonlayer}
\begin{pgfonlayer}{edgelayer}
\draw[] (0.center) to[out= 90, in=-89.9] (1.center);

\end{pgfonlayer}
\end{tikzpicture}
 $=\frac{a}{\sqrt{c + 1}}\begin{pmatrix}
1 & \frac{-1}{c}
\end{pmatrix}$\hfill
\input figs/phaseblack $=\frac{x}{\sqrt{c + 1}}\begin{pmatrix}
1 & 1\\
c & c - (c+1) y
\end{pmatrix}$
\clearpage

\subsubsection[The original ZH-calculus]{The case $a= c = 1$}

Here is the calculus we obtain when $a=c=1$. This is almost the ZH-calculus
\cite{backens2018zh} (black nodes are represented by a white rectangle in ZH), with a slight differences in the parametrization of the
phases:

\input figs/mprodwhite $=\begin{pmatrix}
1 & 0 & 0 & 0 \\
0 & 0 & 0 & 1
\end{pmatrix}$ \hfill
\begin{tikzpicture}[baseline=(current bounding box.center)]
\begin{pgfonlayer}{nodelayer}
\coordinate (0) at (0.25,0);
\node [style=gmonoid] (1) at (0.25,0.5) {};

\end{pgfonlayer}
\begin{pgfonlayer}{edgelayer}
\draw[] (0.center) to[out= 90, in=-89.9] (1.center);

\end{pgfonlayer}
\end{tikzpicture}
 $=\begin{pmatrix}
1 \\
1
\end{pmatrix}$ \hfill
\input figs/coprodwhite $=\begin{pmatrix}
1 & 0 \\
0 & 0 \\
0 & 0 \\
0 & 1
\end{pmatrix}$ \hfill
\begin{tikzpicture}[baseline=(current bounding box.center)]
\begin{pgfonlayer}{nodelayer}
\node [style=gmonoid] (0) at (0.25,0.25) {};
\coordinate (1) at (0.25,0.75);

\end{pgfonlayer}
\begin{pgfonlayer}{edgelayer}
\draw[] (0.center) to[out= 90, in=-89.9] (1.center);

\end{pgfonlayer}
\end{tikzpicture}
 $=\begin{pmatrix}
1 & 1
\end{pmatrix}$\hfill
\input figs/phasewhite  $=x\begin{pmatrix}
1 & 0 \\
0 & y
\end{pmatrix}$

\hfill\begin{tikzpicture}[baseline=(current bounding box.center)]
\begin{pgfonlayer}{nodelayer}
\node [] (0) at (0.5,0) {};
\coordinate (1) at (0.25,0.5);
\coordinate (2) at (0.75,0.5);

\end{pgfonlayer}
\begin{pgfonlayer}{edgelayer}
\draw[] (0.center) to[out=180, in=-89.9] (1.center);
\draw[] (0.center) to[out=  0, in=-89.9] (2.center);

\end{pgfonlayer}
\end{tikzpicture}
 $= \begin{pmatrix}
1 & 0 & 0 & 1
\end{pmatrix}$\hfill
\begin{tikzpicture}[baseline=(current bounding box.center)]
\begin{pgfonlayer}{nodelayer}
\coordinate (0) at (0.25,0);
\coordinate (1) at (0.75,0);
\node [] (2) at (0.5,0.5) {};

\end{pgfonlayer}
\begin{pgfonlayer}{edgelayer}
\draw[] (0.center) to[out= 90, in=180] (2.center);
\draw[] (1.center) to[out= 90, in=  0] (2.center);

\end{pgfonlayer}
\end{tikzpicture}
 $=\begin{pmatrix}
1 \\
0 \\
0 \\
1
\end{pmatrix}$\hfill
\begin{tikzpicture}[baseline=(current bounding box.center)]
\begin{pgfonlayer}{nodelayer}
\coordinate (0) at (0.25,0);
\node [style=hmonoid] (1) at (0.25,0.5) {};
\coordinate (2) at (0.25,1);

\end{pgfonlayer}
\begin{pgfonlayer}{edgelayer}
\draw[] (0.center) to[out= 90, in=-89.9] (1.center);
\draw[] (1.center) to[out= 90, in=-89.9] (2.center);

\end{pgfonlayer}
\end{tikzpicture}
 $=\frac{1}{\sqrt{2}}\begin{pmatrix}
1 & 1 \\
1 & - 1\\
\end{pmatrix}$\hfill

\input figs/mprodblack $=\frac{1}{\sqrt{2}}\begin{pmatrix}
1 & 1 & 1 & 1\\
1 & 1 & 1 & -1\\
\end{pmatrix}$ \hfill
\begin{tikzpicture}[baseline=(current bounding box.center)]
\begin{pgfonlayer}{nodelayer}
\coordinate (0) at (0.25,0);
\node [style=hmonoid] (1) at (0.25,0.5) {};

\end{pgfonlayer}
\begin{pgfonlayer}{edgelayer}
\draw[] (0.center) to[out= 90, in=-89.9] (1.center);

\end{pgfonlayer}
\end{tikzpicture}
 $=\frac{1}{\sqrt{2}}\begin{pmatrix}
1 \\
-1
\end{pmatrix}$ \hfill
\input figs/coprodblack $=\frac{1}{\sqrt{2}}\begin{pmatrix}
1 & 1 \\
1 & 1 \\
1 & 1 \\
1 & -1\\
\end{pmatrix}$ \hfill
\begin{tikzpicture}[baseline=(current bounding box.center)]
\begin{pgfonlayer}{nodelayer}
\node [style=hmonoid] (0) at (0.25,0.25) {};
\coordinate (1) at (0.25,0.75);

\end{pgfonlayer}
\begin{pgfonlayer}{edgelayer}
\draw[] (0.center) to[out= 90, in=-89.9] (1.center);

\end{pgfonlayer}
\end{tikzpicture}
 $=\frac{1}{\sqrt{2}}\begin{pmatrix}
1 & -1
\end{pmatrix}$

\input figs/phaseblack $=\frac{x}{\sqrt{2}}\begin{pmatrix}
1 & 1\\
1 & 1 - 2y\\
\end{pmatrix}$

\subsection{The ZW-calculi}

\subsubsection[The first ZW calculus]{The $Z^{(a,1/c^2)}W_{(c/a,0)}$}

\input figs/mprodwhite $=a\begin{pmatrix}
1 & 0 & 0 & 0 \\
0 & 0 & 0 & \frac{1}{c^{2}}
\end{pmatrix}$ \hfill
\begin{tikzpicture}[baseline=(current bounding box.center)]
\begin{pgfonlayer}{nodelayer}
\coordinate (0) at (0.25,0);
\node [style=gmonoid] (1) at (0.25,0.5) {};

\end{pgfonlayer}
\begin{pgfonlayer}{edgelayer}
\draw[] (0.center) to[out= 90, in=-89.9] (1.center);

\end{pgfonlayer}
\end{tikzpicture}
 $=\frac{1}{a}\begin{pmatrix}
1 \\
c^{2}
\end{pmatrix}$ \hfill
\input figs/coprodwhite $=\begin{pmatrix}
1 & 0 \\
0 & 0 \\
0 & 0 \\
0 & 1
\end{pmatrix}$ \hfill
\begin{tikzpicture}[baseline=(current bounding box.center)]
\begin{pgfonlayer}{nodelayer}
\node [style=gmonoid] (0) at (0.25,0.25) {};
\coordinate (1) at (0.25,0.75);

\end{pgfonlayer}
\begin{pgfonlayer}{edgelayer}
\draw[] (0.center) to[out= 90, in=-89.9] (1.center);

\end{pgfonlayer}
\end{tikzpicture}
 $=\begin{pmatrix}
1 & 1
\end{pmatrix}$\hfill
\input figs/phasewhite  $=x\begin{pmatrix}
1 & 0 \\
0 & y
\end{pmatrix}$

\hfill\begin{tikzpicture}[baseline=(current bounding box.center)]
\begin{pgfonlayer}{nodelayer}
\node [] (0) at (0.5,0) {};
\coordinate (1) at (0.25,0.5);
\coordinate (2) at (0.75,0.5);

\end{pgfonlayer}
\begin{pgfonlayer}{edgelayer}
\draw[] (0.center) to[out=180, in=-89.9] (1.center);
\draw[] (0.center) to[out=  0, in=-89.9] (2.center);

\end{pgfonlayer}
\end{tikzpicture}
 $= a\begin{pmatrix}
1 & 0 & 0 & \frac{1}{c^{2}}
\end{pmatrix}$\hfill
\begin{tikzpicture}[baseline=(current bounding box.center)]
\begin{pgfonlayer}{nodelayer}
\coordinate (0) at (0.25,0);
\coordinate (1) at (0.75,0);
\node [] (2) at (0.5,0.5) {};

\end{pgfonlayer}
\begin{pgfonlayer}{edgelayer}
\draw[] (0.center) to[out= 90, in=180] (2.center);
\draw[] (1.center) to[out= 90, in=  0] (2.center);

\end{pgfonlayer}
\end{tikzpicture}
 $=\frac{1}{a}\begin{pmatrix}
1 \\
0 \\
0 \\
c^{2}
\end{pmatrix}$\hfill
\begin{tikzpicture}[baseline=(current bounding box.center)]
\begin{pgfonlayer}{nodelayer}
\coordinate (0) at (0.25,0);
\node [style=hmonoid] (1) at (0.25,0.5) {};
\coordinate (2) at (0.25,1);

\end{pgfonlayer}
\begin{pgfonlayer}{edgelayer}
\draw[] (0.center) to[out= 90, in=-89.9] (1.center);
\draw[] (1.center) to[out= 90, in=-89.9] (2.center);

\end{pgfonlayer}
\end{tikzpicture}
 $=\begin{pmatrix}
0 & \frac{1}{c} \\
c & 0
\end{pmatrix}$\hfill

\input figs/mprodblack $=\begin{pmatrix}
0 & \frac{1}{c} & \frac{1}{c} & 0 \\
c & 0 & 0 & 0
\end{pmatrix}$ \hfill
\begin{tikzpicture}[baseline=(current bounding box.center)]
\begin{pgfonlayer}{nodelayer}
\coordinate (0) at (0.25,0);
\node [style=hmonoid] (1) at (0.25,0.5) {};

\end{pgfonlayer}
\begin{pgfonlayer}{edgelayer}
\draw[] (0.center) to[out= 90, in=-89.9] (1.center);

\end{pgfonlayer}
\end{tikzpicture}
 $=\begin{pmatrix}
0 \\
c
\end{pmatrix}$ \hfill
\input figs/coprodblack $=\frac{1}{a}\begin{pmatrix}
0 & \frac{1}{c} \\
c & 0 \\
c & 0 \\
0 & 0
\end{pmatrix}$ \hfill
\begin{tikzpicture}[baseline=(current bounding box.center)]
\begin{pgfonlayer}{nodelayer}
\node [style=hmonoid] (0) at (0.25,0.25) {};
\coordinate (1) at (0.25,0.75);

\end{pgfonlayer}
\begin{pgfonlayer}{edgelayer}
\draw[] (0.center) to[out= 90, in=-89.9] (1.center);

\end{pgfonlayer}
\end{tikzpicture}
 $=a\begin{pmatrix}
0 & \frac{1}{c}
\end{pmatrix}$\hfill
\input figs/phaseblack $=x\begin{pmatrix}
c y & \frac{1}{c} \\
c & 0
\end{pmatrix}$

\clearpage
\subsubsection[The original ZW calculus]{The case $a = c = 1$}

This is the ZW-calculus of \cite{hadzihasanovic2017algebra,HNW}:

\input figs/mprodwhite $=\begin{pmatrix}
1 & 0 & 0 & 0 \\
0 & 0 & 0 & 1
\end{pmatrix}$ \hfill
\begin{tikzpicture}[baseline=(current bounding box.center)]
\begin{pgfonlayer}{nodelayer}
\coordinate (0) at (0.25,0);
\node [style=gmonoid] (1) at (0.25,0.5) {};

\end{pgfonlayer}
\begin{pgfonlayer}{edgelayer}
\draw[] (0.center) to[out= 90, in=-89.9] (1.center);

\end{pgfonlayer}
\end{tikzpicture}
 $=\begin{pmatrix}
1 \\
1
\end{pmatrix}$ \hfill
\input figs/coprodwhite $=\begin{pmatrix}
1 & 0 \\
0 & 0 \\
0 & 0 \\
0 & 1
\end{pmatrix}$ \hfill
\begin{tikzpicture}[baseline=(current bounding box.center)]
\begin{pgfonlayer}{nodelayer}
\node [style=gmonoid] (0) at (0.25,0.25) {};
\coordinate (1) at (0.25,0.75);

\end{pgfonlayer}
\begin{pgfonlayer}{edgelayer}
\draw[] (0.center) to[out= 90, in=-89.9] (1.center);

\end{pgfonlayer}
\end{tikzpicture}
 $=\begin{pmatrix}
1 & 1
\end{pmatrix}$\hfill
\input figs/phasewhite  $=x\begin{pmatrix}
1 & 0 \\
0 & y
\end{pmatrix}$

\hfill\begin{tikzpicture}[baseline=(current bounding box.center)]
\begin{pgfonlayer}{nodelayer}
\node [] (0) at (0.5,0) {};
\coordinate (1) at (0.25,0.5);
\coordinate (2) at (0.75,0.5);

\end{pgfonlayer}
\begin{pgfonlayer}{edgelayer}
\draw[] (0.center) to[out=180, in=-89.9] (1.center);
\draw[] (0.center) to[out=  0, in=-89.9] (2.center);

\end{pgfonlayer}
\end{tikzpicture}
 $=\begin{pmatrix}
1 & 0 & 0 & 1
\end{pmatrix}$\hfill
\begin{tikzpicture}[baseline=(current bounding box.center)]
\begin{pgfonlayer}{nodelayer}
\coordinate (0) at (0.25,0);
\coordinate (1) at (0.75,0);
\node [] (2) at (0.5,0.5) {};

\end{pgfonlayer}
\begin{pgfonlayer}{edgelayer}
\draw[] (0.center) to[out= 90, in=180] (2.center);
\draw[] (1.center) to[out= 90, in=  0] (2.center);

\end{pgfonlayer}
\end{tikzpicture}
 $=\begin{pmatrix}
1 \\
0 \\
0 \\
1
\end{pmatrix}$\hfill
\begin{tikzpicture}[baseline=(current bounding box.center)]
\begin{pgfonlayer}{nodelayer}
\coordinate (0) at (0.25,0);
\node [style=hmonoid] (1) at (0.25,0.5) {};
\coordinate (2) at (0.25,1);

\end{pgfonlayer}
\begin{pgfonlayer}{edgelayer}
\draw[] (0.center) to[out= 90, in=-89.9] (1.center);
\draw[] (1.center) to[out= 90, in=-89.9] (2.center);

\end{pgfonlayer}
\end{tikzpicture}
 $=\begin{pmatrix}
0 & 1 \\
1 & 0
\end{pmatrix}$\hfill

\input figs/mprodblack $=\begin{pmatrix}
0 & 1 & 1 & 0 \\
1 & 0 & 0 & 0
\end{pmatrix}$ \hfill
\begin{tikzpicture}[baseline=(current bounding box.center)]
\begin{pgfonlayer}{nodelayer}
\coordinate (0) at (0.25,0);
\node [style=hmonoid] (1) at (0.25,0.5) {};

\end{pgfonlayer}
\begin{pgfonlayer}{edgelayer}
\draw[] (0.center) to[out= 90, in=-89.9] (1.center);

\end{pgfonlayer}
\end{tikzpicture}
 $=\begin{pmatrix}
0 \\
1
\end{pmatrix}$ \hfill
\input figs/coprodblack $=\begin{pmatrix}
0 & 1 \\
1 & 0 \\
1 & 0 \\
0 & 0
\end{pmatrix}$ \hfill
\begin{tikzpicture}[baseline=(current bounding box.center)]
\begin{pgfonlayer}{nodelayer}
\node [style=hmonoid] (0) at (0.25,0.25) {};
\coordinate (1) at (0.25,0.75);

\end{pgfonlayer}
\begin{pgfonlayer}{edgelayer}
\draw[] (0.center) to[out= 90, in=-89.9] (1.center);

\end{pgfonlayer}
\end{tikzpicture}
 $=\begin{pmatrix}
0 & 1
\end{pmatrix}$\hfill
\input figs/phaseblack $=x\begin{pmatrix}
y & 1 \\
1 & 0
\end{pmatrix}$

\subsubsection[The second ZW calculus]{The $W^{(a,0)}Z_{(b/a,1/b^2)}$ calculus}

This is very similar to the previous calculus, except that $W$ is now chosen as
the white node, meaning that the black node is actually $Z$, up to the dualizer.

\input figs/mprodwhite $=a\begin{pmatrix}
1 & 0 & 0 & 0 \\
0 & 1 & 1 & 0
\end{pmatrix}$ \hfill
\begin{tikzpicture}[baseline=(current bounding box.center)]
\begin{pgfonlayer}{nodelayer}
\coordinate (0) at (0.25,0);
\node [style=gmonoid] (1) at (0.25,0.5) {};

\end{pgfonlayer}
\begin{pgfonlayer}{edgelayer}
\draw[] (0.center) to[out= 90, in=-89.9] (1.center);

\end{pgfonlayer}
\end{tikzpicture}
 $=\frac{1}{a}\begin{pmatrix}
1 \\
0
\end{pmatrix}$ \hfill
\input figs/coprodwhite $=\begin{pmatrix}
0 & 0 \\
1 & 0 \\
1 & 0 \\
0 & 1
\end{pmatrix}$ \hfill
\begin{tikzpicture}[baseline=(current bounding box.center)]
\begin{pgfonlayer}{nodelayer}
\node [style=gmonoid] (0) at (0.25,0.25) {};
\coordinate (1) at (0.25,0.75);

\end{pgfonlayer}
\begin{pgfonlayer}{edgelayer}
\draw[] (0.center) to[out= 90, in=-89.9] (1.center);

\end{pgfonlayer}
\end{tikzpicture}
 $=\begin{pmatrix}
0 & 1
\end{pmatrix}$\hfill
\input figs/phasewhite  $=x\begin{pmatrix}
1 & 0 \\
y & 1
\end{pmatrix}$

\hfill\begin{tikzpicture}[baseline=(current bounding box.center)]
\begin{pgfonlayer}{nodelayer}
\node [] (0) at (0.5,0) {};
\coordinate (1) at (0.25,0.5);
\coordinate (2) at (0.75,0.5);

\end{pgfonlayer}
\begin{pgfonlayer}{edgelayer}
\draw[] (0.center) to[out=180, in=-89.9] (1.center);
\draw[] (0.center) to[out=  0, in=-89.9] (2.center);

\end{pgfonlayer}
\end{tikzpicture}
 $= a\begin{pmatrix}
0 & 1 & 1 & 0
\end{pmatrix}$\hfill
\begin{tikzpicture}[baseline=(current bounding box.center)]
\begin{pgfonlayer}{nodelayer}
\coordinate (0) at (0.25,0);
\coordinate (1) at (0.75,0);
\node [] (2) at (0.5,0.5) {};

\end{pgfonlayer}
\begin{pgfonlayer}{edgelayer}
\draw[] (0.center) to[out= 90, in=180] (2.center);
\draw[] (1.center) to[out= 90, in=  0] (2.center);

\end{pgfonlayer}
\end{tikzpicture}
 $=\frac{1}{a}\begin{pmatrix}
0 \\
1 \\
1 \\
0
\end{pmatrix}$\hfill
\begin{tikzpicture}[baseline=(current bounding box.center)]
\begin{pgfonlayer}{nodelayer}
\coordinate (0) at (0.25,0);
\node [style=hmonoid] (1) at (0.25,0.5) {};
\coordinate (2) at (0.25,1);

\end{pgfonlayer}
\begin{pgfonlayer}{edgelayer}
\draw[] (0.center) to[out= 90, in=-89.9] (1.center);
\draw[] (1.center) to[out= 90, in=-89.9] (2.center);

\end{pgfonlayer}
\end{tikzpicture}
 $=\begin{pmatrix}
0 & b \\
\frac{1}{b} & 0
\end{pmatrix}$\hfill

\input figs/mprodblack $=\begin{pmatrix}
0 & 0 & 0 & b \\
\frac{1}{b} & 0 & 0 & 0
\end{pmatrix}$ \hfill
\begin{tikzpicture}[baseline=(current bounding box.center)]
\begin{pgfonlayer}{nodelayer}
\coordinate (0) at (0.25,0);
\node [style=hmonoid] (1) at (0.25,0.5) {};

\end{pgfonlayer}
\begin{pgfonlayer}{edgelayer}
\draw[] (0.center) to[out= 90, in=-89.9] (1.center);

\end{pgfonlayer}
\end{tikzpicture}
 $=\begin{pmatrix}
b \\
\frac{1}{b}
\end{pmatrix}$ \hfill
\input figs/coprodblack $=\frac{1}{a}\begin{pmatrix}
0 & b \\
0 & 0 \\
0 & 0 \\
\frac{1}{b} & 0
\end{pmatrix}$ \hfill
\begin{tikzpicture}[baseline=(current bounding box.center)]
\begin{pgfonlayer}{nodelayer}
\node [style=hmonoid] (0) at (0.25,0.25) {};
\coordinate (1) at (0.25,0.75);

\end{pgfonlayer}
\begin{pgfonlayer}{edgelayer}
\draw[] (0.center) to[out= 90, in=-89.9] (1.center);

\end{pgfonlayer}
\end{tikzpicture}
 $=a\begin{pmatrix}
\frac{1}{b} & b
\end{pmatrix}$\hfill
\input figs/phaseblack $=x\begin{pmatrix}
0 & b y \\
\frac{1}{b} & 0
\end{pmatrix}$

\end{document}